\documentclass[12pt]{article}
\usepackage{amsmath,amsfonts,amsthm}
\usepackage{algorithmic}
\usepackage{algorithm}
\usepackage{array}
\usepackage[caption=false,font=normalsize,labelfont=sf,textfont=sf]{subfig}
\usepackage{textcomp,color}
\usepackage{stfloats}
\usepackage{url}
\usepackage{verbatim}
\usepackage{graphicx}
\usepackage{cite}
\allowdisplaybreaks
\usepackage{amssymb}
\usepackage{tikz} 
\usetikzlibrary{arrows,shapes,chains}

\newtheorem{assumption}{Assumption}
\newtheorem{lemma}{Lemma}
\newtheorem{problem}{Problem}
\newtheorem{theorem}{Theorem}
\newtheorem{remark}{Remark}
\newtheorem{corollary}{Corollary}

\newcommand{\tM}{\textrm}
\newcommand{\tetrahedron}{
\mathchoice
{\includegraphics[height=1.6ex]{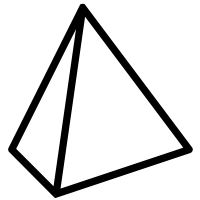}} % \displaystyle
{\includegraphics[height=1.6ex]{tetrahedron}} % \textstyle
{\includegraphics[height=1.3ex]{tetrahedron}} % \scriptstyle
{\includegraphics[height=0.9ex]{tetrahedron}} % \scriptscriptstyle
}

\begin{document}

\title{Distributed Finite-Time Cooperative Localization for Three-Dimensional Sensor Networks}

\author{Jinze Wu, Lorenzo Zino, Zhiyun Lin, and Alessandro Rizzo
\thanks{J. Wu, L. Zino, and A. Rizzo are with the Department of Electronics and Telecommunications, Politecnico di Torino, Torino, Italy. Emails: \texttt{\{jinze.wu,lorenzo.zino,alessandro.rizzo\}@polito.it}. J.~Wu and Z. Lin are with Shenzhen Key Laboratory of Control Theory and Intelligent Systems, School of System Design and Intelligent Manufacturing, Southern University of Science and Technology, Shenzhen, China. Email: \texttt{linzy@sustech.edu.cn}. Z. Lin is also with Peng Cheng Laboratory, Shenzhen, China. A. Rizzo is also with the Institute for Invention, Innovation, and Entrepreneurship, New York University Tandon School of Engineering, Brooklyn NY, US. }
}% <-this % stops a space
%\thanks{Manuscript received April 19, 2021; revised August 16, 2021.}}


\maketitle

\begin{abstract}
This paper addresses the distributed localization problem for a network of sensors placed in a three-dimensional space, in which sensors are able to perform range measurements, i.e., measure the relative distance between them, and exchange information on a network structure. First, we derive a necessary and sufficient condition for node localizability using barycentric coordinates. Then, building on this theoretical result, we design a distributed localizability verification algorithm, in which we propose and employ a novel distributed finite-time algorithm for sum consensus. Finally, we develop a distributed localization algorithm based on conjugate gradient method, and we derive a theoretical guarantee on its performance, which ensures finite-time convergence to the exact position for all localizable nodes. The efficiency of our algorithm compared to the existing ones from the state-of-the-art literature is further demonstrated through numerical simulations.
\end{abstract}

%\begin{IEEEkeywords}
%Distributed localization, sensor network, node localizability, conjugate gradient.
%\end{IEEEkeywords}

%——————————%%%%%%%——————————
\section{Introduction} \label{Sec: Introduction}

Localization has become a critical issue as the position information of a set of agents plays an increasingly important role in multiple applications, such as target searching, environmental monitoring, and formation control~\cite{patwari2005locating, sayed2005network, lowell2011military, bayat2017environmental, lin2014distributed}. The most common method to obtain position information is to equip all agents with GPS, a satellite-based global positioning system that provides the coordinates of receivers~\cite{el2002introduction}. However, GPS has some drawbacks in terms of coverage and power consumption~\cite{buehrer2018collaborative}. More specifically, it may not work properly in environments with obstructions between the GPS satellites and the receivers, e.g., interior buildings and areas with dense vegetation or mountains. Furthermore, a scenario in which all agents are equipped with GPS consumes more energy than one with only a few of them equipped with GPS. Thus, cooperative localization based on terrestrial techniques such as cellular, WiFi, and ultra-wideband has come to the fore.

The objective of cooperative localization is to determine the Euclidean coordinates of all agents given the Euclidean coordinates of a (small) subset of reference agents, and assuming that the rest of the agents is able to make local measurements (e.g., the relative distance, bearing, or angle between them) and exchange information over a network. For this reason, in the rest of this paper we will refer to the set of agents as a \emph{sensor network}. Given the importance of this problem, several approaches have been proposed in the literature.

Since the essence of the cooperative localization problem is nonlinear (being the relationship between Euclidean coordinates and local measurements generally nonlinear), most authors in the literature focused on nonlinear algorithms, formulating the cooperative localization problem as a constrained nonlinear optimization problem~\cite{hendrickson1995molecule, aspnes2006theory, mao2007wireless, jing2021angle}. However, these approaches have several crucial limitations. First, solvers for such nonlinear optimization problems typically lack in theoretical guarantees to ensure global convergence~\cite{aspnes2006theory}. Second, as the scale of the sensor network grows, solving nonlinear equations results in an increased number of saddle points due to iterations, limiting the possibility to implement them in large-scale real-world applications~\cite{helmke2013equivariant}. Third, most of the nonlinear localization algorithms proposed in the literature need to be solved in a centralized fashion, yielding a more complex and less robust architecture. 

In order to address the limitations of nonlinear algorithms, distributed linear cooperative localization algorithms have attracted a lot of research interest in recent years. In particular, different algorithms have been proposed to tackle this problem, depending on the type of local measurements that the sensors are able to perform, including range-based~\cite{khan2009distributed, diao2015sequential, diao2014barycentric, cheng2016single, han2017barycentric}, bearing-based~\cite{zhao2016localizability, lin2016distributed, li2019globally, cao2021bearing}, relative position-based~\cite{barooah2007estimation, lin2015distributed}, angle-based~\cite{lin2020distributed, chen2022triangular, chen20223} and mixed measurements-based~\cite{lin2017distributed, fang20203, fang2021angle} methods. A common technique to represent space in these effort is the use of barycentric coordinates~\cite{möbius2018barycentrische} --a coordinate system in which the location of a point is specified with reference to other points. Barycentric coordinate systems have emerged as one of the most effective tools to obtain the prerequisites for addressing the cooperative localization problem in a distributed and linear manner.

The authors in~\cite{khan2009distributed} proposed a distributed localization algorithm based on barycentric coordinates that uses range-based measurements (i.e., in which it is assumed that sensors can measure the relative distance between them). This algorithm is implemented by expressing the positions of the sensors as a pseudo linear system with all nonlinearities hidden in range measurements and relies on the assumption that each sensor requires to lie inside the convex hull of its neighbors. To remove this assumption,~\cite{diao2014barycentric} developed a novel localization algorithm in a two-dimensional space. Building on this approach,~\cite{cheng2016single} presented a more robust localization algorithm in a two-dimensional space, adopting the idea of the congruent framework. Such approach was successfully extended to a three-dimensional space in~\cite{han2017barycentric}. Besides range-based measurements, barycentric coordinates were also employed to solve localization problems in other scenarios, with bearing~\cite{lin2016distributed, cao2021bearing}, relative position~\cite{lin2015distributed}, angle~\cite{lin2020distributed} and mixed~\cite{lin2017distributed} measurements. Moreover, the angle-based localization algorithms proposed in~\cite{chen2022triangular, chen20223} can be converted into barycentric coordinates for solution as well.

In addition to designing localization algorithms, a problem of paramount importance is to determine which nodes are localizable, given the specific characteristics of the sensor network setup (i.e., topology of the sensing and communication channel, and type of measurement that sensor can perform)~\cite{ping2022node}. Most of the aforementioned works have focused on proposing localization algorithms under the assumption that all sensor nodes can be localized in the sensor network. In other words, these works solely address network localizability without delving into the node localizability. However, research on node localizability in the localization problem based on barycentric coordinates holds significant importance. Neglecting whether a sensor node is localizable or not can lead to errors in the localization process. For instance, the location of a localizable sensor may be inaccurately estimated due to the iterative propagation of wrong positions from unlocalizable sensors. Therefore, it is crucial to identify all localizable and unlocalizable sensor nodes before implementing a distributed localization algorithm. Accurate identification of localizable sensors is a key step to ensure the reliability and precision of the localization process in a sensor network.

In this paper, we address the localization problem, without making any a priori assumption on localizability. Specifically, we deal with the range-based localization problem for sensor networks in a three-dimensional space. The main contribution of this paper is threefold. First, we prove a necessary and sufficient condition for node localizability of range-based localization problem in a three-dimensional space using barycentric coordinates, and then propose a distributed verification algorithm to distinguish between localizable and unlocalizable sensor nodes in the sensor network. With this verification and filtering process, many localization algorithms based on barycentric coordinates can still succeed even in the presence of unlocalizable sensor nodes, further lifting the restrictions of localization algorithms. Second, we develop a distributed conjugate gradient localization algorithm to efficiently solve the localization problem for the sensor network in finite time. Its finite-time convergence is proved theoretically, and then validated in numerical simulations. Third, in the development of our algorithms, we propose a novel distributed algorithm, which is able to achieve sum consensus among all nodes in the network in finite time. Such algorithm, which is used in our localization process, is a distributed consensus algorithm that considers the tradeoff between memory size and iteration steps, and can be directly employed in other application fields.

The rest of the paper is organized as follows. Section~\ref{Sec: Preliminaries} presents notation and preliminaries. In Section~\ref{Sec: Formulation}, we formulate the research problem. In Section~\ref{Sec: Localizability}, we propose a distributed algorithm to address the localizability verification problem. Section~\ref{Sec: Localization} is devoted to the distributed localization algorithm and the proof of its finite-time convergence. Numerical simulations are presented in Section~\ref{Sec: Simulation}. Section~\ref{Sec: Conclusion} concludes the paper and outlines future research directions.

%——————————%%%%%%%——————————
\section{Notation and Preliminaries} \label{Sec: Preliminaries}

%In this section, we present the notation used in the paper and we review some important facts about algebraic graph theory~\cite{biggs1993algebraic,godsil2001algebraic} and barycentric coordinates~\cite{möbius2018barycentrische}. 

%%%%%%%
\subsection{Notation}

We use uppercase letters for matrices, bold lowercase letters for vectors, and lowercase letters for scalars. With $\mathbf{x}\in\mathbb R^n$ we refer to a column vector, and we use the transpose operator ($\mathbf{x}^\top$) to denote row vectors. Let $\mathbb{R}^n$ be the $n$-dimensional real coordinate space. $I_n$ denotes the identity matrix of order $n \times n$, while $\mathbf{1}_n$ denotes the $n$-dimensional vector with all entries equal to $1$ and $\mathbf{e_i}$ denotes a vector whose $i$th entry is $1$ and all other entries are $0$. The symbol $\| \cdot \|$ denotes the 2-norm and $\otimes$ denotes the Kronecker product. Moreover, the symbol $\lceil x \rceil $ refers to the smallest integer that is larger than or equal to a real number $x$. Finally, let $\tM{diag}( \cdot )$ denote the diagonal matrix. Given a matrix $A\in\mathbb R^{n\times m}$, $\tM{rank}( A )$ and $\tM{ker}( A )$ are its rank and kernel, respectively.

%%%%%%%
\subsection{Graph Theory}

An {\em undirected graph} $\mathcal{G} = (\mathcal{V}, \mathcal{E})$ consists of a nonempty set $\mathcal{V}$, whose elements are called {\em nodes} or {\em vertices}, and a set of unordered pairs of nodes $\mathcal{E} \subseteq \mathcal{V}\times \mathcal{V}$, whose elements are called {\em edges}. For a node $i$, we define its {\em neighbor set} as $\mathcal{N}_i:= \{j \in \mathcal{V} : (i,j) \in \mathcal{E}\}$. Clearly, if $i \in \mathcal{N}_j$ then $j \in \mathcal{N}_i$. A {\em path} of length $\ell$ from $i\in\mathcal V$ to $j\in\mathcal V$ is a sequence of $\ell$ edges of the form $(i,v_1),(v_1,v_2),\dots,(v_{\ell-1},j)$. A graph is said to be \emph{connected} if there is a path between every pair of vertices. The {\em diameter} of a connected graph, denoted by $\delta$, is the maximum among all the lengths of the shortest path between any pair of vertices of the graph. A subset of $d$ vertices such that each and every pair of them is connected through an edge is said to be a \emph{clique} of order $d$. Here, we shall refer to cliques of order $4$ as \emph{tetrahedrons} and use the notation $\tetrahedron ijkl$ for a tetrahedron with vertices $i,j,k,l\in\mathcal V$. Note that a clique of order $5$ is formed by $5$ tetrahedrons.

Given a set of $n$ nodes positioned in the Euclidean space $\mathbb R^3$, we define their \emph{configuration} as the vector $\mathbf{p} = [\mathbf{p}^\top_1, \cdots, \mathbf{p}^\top_n]^\top \in \mathbb R^{3n}$, where the entry $\mathbf{p_i} \in \mathbb R^3$ is the position of node $i$ in the three-dimensional Euclidean space. A configuration $\mathbf{p} = [\mathbf{p}^\top_1, \cdots, \mathbf{p}^\top_n]^\top \in \mathbb R^{3n}$ is said to be {\em generic} if the coordinates $\mathbf{p_1}, \cdots, \mathbf{p_n}$ do not satisfy any nonzero polynomial equation with integer coefficients or equivalently algebraic coefficients~\cite{singer2010uniqueness}. In other words, a generic configuration has no degeneracy, that is, no three points staying on the same line, no three lines go through the same point, no four points staying on the same plane, etc.

A {\em framework} $\mathcal{F} = (\mathcal{G}, \mathbf{p})$ in the Euclidean space $\mathbb R^3$ is formed by a graph $\mathcal{G} = (\mathcal{V}, \mathcal{E})$ and a configuration $\mathbf{p}\in\mathbb R^{3n}$. Two frameworks with the same graph $(\mathcal{G}, \mathbf{p})$ and $(\mathcal{G}, \mathbf{q})$ are said to be {\em congruent} (denoted by $(\mathcal G, \mathbf{p}) \equiv (\mathcal G, \mathbf{q})$) if and only if (iff) $\|\mathbf{p_i} - \mathbf{p_j} \| = \|\mathbf{q_i} - \mathbf{q_j}\|$ holds for all pairs $i, j \in \mathcal{V}$. In other words, all nodes of two congruent frameworks maintain the same distances between pairs. It is easy to check that the property of being congruent is an equivalence relation. It is important to notice that such an equivalence relation preserves volumes in the space.

%%%%%%%
\subsection{Barycentric Coordinates}

The {\em barycentric coordinates} were introduced by A. F. Möbius in 1827 as mass points to define a coordinate-free geometry~\cite{möbius2018barycentrische}. In fact, barycentric coordinates characterize the relative position of a node with respect to other nodes. Specifically, given five nodes $i, j, k, l$, and $h$ with their Euclidean coordinates $\mathbf{p_i}, \mathbf{p_j}, \mathbf{p_k}, \mathbf{p_l},\mathbf{p_h}\in\mathbb R^3$, the barycentric coordinates of node $i$ with respect to nodes $j, k, l$, and $h$ are equal to the quadruple $\{ a_{ij}, a_{ik}, a_{il}, a_{ih} \}$, solution of the following system of equations:
\begin{equation}
\label{Eq: barycentric coordinates}
\begin{cases}
	\mathbf{p_i} = a_{ij} \mathbf{p_j} + a_{ik} \mathbf{p_k} + a_{il} \mathbf{p_l} + a_{ih} \mathbf{p_h},\\
	a_{ij} + a_{ik} + a_{il} +a_{ih} = 1.
\end{cases}
\end{equation}
In other words, using the barycentric coordinates, we represent the position of a point in a three-dimensional space as a combination of the positions of four other points in the space, and we use the weights of such combination to identify the point. Fig.~\ref{Fig: barycentric coordinate} illustrates a graphic representation of the barycentric coordinates. 

The barycentric coordinates of node $i$ can be calculated using the signed volumes between corresponding tetrahedrons~\cite{möbius2018barycentrische}. Specifically, it holds
\begin{equation}
\label{Eq: barycentric coordinates - a}
a_{i j}=\frac{V_{i k l h}}{V_{j k l h}}, 
a_{i k} = \frac{V_{j i l h}}{V_{j k l h}}, \\
a_{i l}=\frac{V_{j k i h}}{V_{j k l h}},
a_{i h}=\frac{V_{j k l i}}{V_{j k l h}},
\end{equation}
where $V_{iklh}, V_{jilh}, V_{jkih}, V_{jkli}$ and $V_{jklh}$ are the signed volumes of the tetrahedrons $\tetrahedron iklh$, $\tetrahedron jilh$, $\tetrahedron jkih$, $\tetrahedron jkli$, and $\tetrahedron jklh$, respectively. The signed volume $V_{iklh}$ is equal in modulus to the volume of the tetrahedron $\tetrahedron iklh$, with positive sign if nodes $i, k, l$, and $h$ conform to the right hand corkscrew rule, and negative otherwise. Given the coordinates of its four vertices, the signed volume $V_{iklh}$ is equal to the following determinant of a $4\times4$ matrix:
\begin{equation}
\label{Eq: volume}
	V_{iklh} = \frac{1}{6} \left| 
	\begin{array}{cccc}
		1 & 1 & 1 & 1 \\
		\mathbf{p_{i}} & \mathbf{p_{k}} & \mathbf{p_{l}} & \mathbf{p_{h}} \\
	\end{array} \right|.
\end{equation}
Note that, due to the volume conservation law under roto-translation, the barycentric coordinates of a node are invariant with respect to congruent frameworks. 

\begin{figure}
	\centering
	\includegraphics[width=0.5 \linewidth]{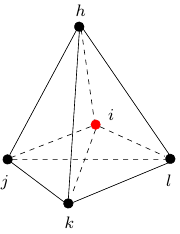}
	\caption{Illustration of the barycentric coordinates of node $i$ with respect to nodes $j, k, l$, and $h$ in a three-dimensional space. }\vspace{-10pt}
	\label{Fig: barycentric coordinate}
\end{figure}

%——————————%%%%%%%——————————
\section{Problem Setup and Statement} \label{Sec: Formulation}

%In this section, we describe the setup of the localization problem and formulate the problem statement.

%%%%%%%
\subsection{Problem Setup}

Consider $n$ sensors positioned in a three-dimensional Euclidean space. Sensors are allowed to interact with each others ---measuring their relative distance (range measurement)--- and communicate, subject to some constraints. Specifically, we define a \emph{sensor network} that is represented by an undirected graph $\mathcal{G} = (\mathcal{V}, \mathcal{E})$ by associating a node of the graph with each sensor, and edges to describe the sensing and communication topology of the sensor network. That is, we assume that $(i,j)\in\mathcal E$ iff sensors $i$ and $j$ can measure their relative distance and communicate, i.e., we assume that the sensing and communication graphs of the sensor network are the same and that the network is connected. 

The sensor network consists of a small number of sensors with known position and a large number of sensors with unknown position. The sensors whose position are known (e.g., obtained through GPS or being able to self-calibrate) are called \emph{anchor nodes}. The rest of the sensors are called \emph{free nodes}. The objective of cooperative localization is to estimate the Euclidean coordinates in the global coordinate system of the free nodes in a distributed fashion by exchanging information on their range measurements, given the absolute Euclidean coordinates of the anchor nodes. %To this aim, nodes can make measurements of their relative distance (range measurements) and they can communicate with their neighbors on the network.

%In this paper, we study a sensor network in a three-dimensional space with $n$ sensor nodes. 
We observe that a necessary condition for all free nodes to be localized in a three-dimensional space using relative measurements is that the sensor network has at least four anchor nodes~\cite{eren2004rigidity}. Hence, without any loss in generality, we assume there are four anchor nodes and $n - 4$ free nodes. We denote these sets as $\mathcal V = \mathcal V_f \cup \mathcal V_a$, with $\mathcal V_f =\{1, \cdots, n-4\}$ and $\mathcal V_a = \{n-3, \cdots, n\}$ denoting the free node set and the anchor node set, respectively. If the sensor network has the edge $(i, j) \in \mathcal{E}$, then node $i$ can measure the range measurement $d_{ij} = \Vert \mathbf{p_i} - \mathbf{p_j} \Vert$ and exchange such information with node $j$. Let $\mathbf{p_f} = [\mathbf{p_1}^\top, \cdots, \mathbf{p_{n-4}}^\top]^\top\in\mathbb{R}^{3(n-4)}$ and $\mathbf{p_a} = [\mathbf{p_{n-3}}^\top, \cdots, \mathbf{p_n}^\top]^\top\in\mathbb{R}^{12}$ denote the Euclidean coordinate vectors of the free nodes and of the anchor nodes, respectively. It is clear that $\mathbf{p_a}$ is known, but the Euclidean coordinates of free node set $\mathbf{p_f}$ needs to be estimated.

All free nodes in the network can be characterized depending on their possibility to be localized in a distributed fashion using barycentric coordinates. Hence, we define two different classes: \emph{localizable} nodes and \emph{unlocalizable} nodes. According to~\cite{eren2004rigidity}, a node requires at least four neighbors for localizability in a three-dimensional space. Moreover, in order to use barycentric coordinates, a node $i\in\mathcal V$ should position itself with respect to other $4$ nodes. Intuitively, we need these $5$ nodes (i.e., node $i$ and the four neighbors) to be able to communicate in order to compute the volumes of the $5$ tetrahedrons they form and, ultimately, the barycentric coordinates of $i$. Hence, node $i$ should belong to a clique of order $5$ to be localized. As a consequence, all nodes that do not belong to any clique of order $5$ belong necessarily to the class of unlocalizable nodes, which we term \emph{scarcely connected} nodes.

\begin{remark} \label{rem:clique}
Node $i$ has information on its neighbor set $\mathcal{N}_i$. Moreover, $i$ can exchange such information with its neighbors. Hence, for any $j \in \mathcal{N}_i$, the neighbor set $\mathcal{N}_j$ can be made available to node $i$ via communication. Thus, by comparing these sets, it is easy for node $i$ to know whether it belongs to a clique of order $5$, without needing centralized information.
\end{remark}

For the sake of simplicity, we will assume that all scarcely connected nodes are detected a priori following Remark~\ref{rem:clique} and removed from the network. Hence, we assume that each of all the (remaining) $n - 4$ free nodes belong to at least one clique of order $5$. However, belonging to a clique of order at least $5$ is a necessary condition for localizability, but is not sufficient. Hence, we need to determine a sufficient condition for localizability, and then design an algorithm to classify the free nodes depending on their localizability, before any localization algorithm can be applied.

%%%%%%%
\subsection{Problem Statement}

In view of our problem setting, the distributed cooperative localization problem explored in this paper is described as follows. We consider a sensor network $\mathcal{G}$ with $n$ nodes in a three-dimensional space, in which sensors can make range measurements and communicate with their neighbors on $\mathcal{G}$. We assume that the absolute positions of four anchor nodes $\mathbf{p_a}$ are known. Our objective is twofold.

\begin{problem}[Localizability verification] \label{problem1}
    Design a distributed localizability test scheme to determine for which of the free nodes it is possible to determine the absolute position.
\end{problem} 

\begin{problem}[Localization algorithm] \label{problem2}
    Design a localization algorithm to estimate $\mathbf{p_i}$ in finite time, for each free node $i\in\mathcal V_f$ that is localizable. 
\end{problem}

%——————————%%%%%%%——————————
\section{Localizability Verification Algorithm}  \label{Sec: Localizability}

In this section, we address Problem~\ref{problem1}. We start by providing a necessary and sufficient condition for node localizability. Then, we use such condition to build an algorithm to verify node localizability in a distributed fashion.

%%%%%%%
\subsection{Necessary and Sufficient Node Localizability Condition} \label{SubSec: Node Localizability Condition}

Before presenting the necessary and sufficient condition of node localizability, we first illustrate how the problem of estimating the Euclidean coordinates of a sensor network can be cast as a set of linear equations that can be derived in a distributed fashion using barycentric coordinates.

Consider a clique of order $5$ formed by nodes $i,j,k,l,h\in\mathcal V$, with positions $\mathbf{p_i},\mathbf{p_j},\mathbf{p_k},\mathbf{p_l},\mathbf{p_h}\in\mathbb R^3$ in a framework we denote as $\mathcal{F}_p$. Node $i$ can communicate with the other four neighbors and can get access to the range measurements $d_{ij}, d_{ik}, d_{il}$, and $d_{ih}$ through onboard sensing or communication from its neighbors. Moreover, since communication and sensing networks coincide, node $i$ can also receive direct information about measurements $d_{jk}, d_{jl}, d_{jh}, d_{kl}, d_{kh}$, and $d_{lh}$ through communication with nodes $j, k, l$, and $h$. These measurements can be utilized to construct a matrix $D\in \mathbb{R}^{5\times 5}$ with its generic entry $D_{xy}:=d_{xy}^2$, that is, equal to the square of the distance measured between nodes $x$ and $y$ for $x, y \in \{i, j, k, l, h\}$, and we define $J = I_5 - \frac{1}{5} \mathbf{1}_5 \mathbf{1}_5^\top$. Then we can apply the Congruent Framework Construction (CFC) algorithm summarized in the pseudocode in Algorithm~\ref{Alg: Congruent Framework Construction Algorithm}, obtaining matrix $Q = [\mathbf{q_i}\,\mathbf{q_j}\,\mathbf{q_k} \,\mathbf{q_l}\,\mathbf{q_h}]$, to construct a congruent framework $\mathcal{F}_q$ for the clique. 

\begin{algorithm}[H]
	\caption{CFC Algorithm}
	\label{Alg: Congruent Framework Construction Algorithm}
	\renewcommand{\algorithmicrequire}{\textbf{$\bullet$}}
	\begin{algorithmic}[1]
		\REQUIRE \textbf{Function:} $Q = \textbf{CFC}(D)$
		%\REQUIRE \textbf{Input:} Squared distance matrix $D$
		%\REQUIRE \textbf{Output:} Euclidean coordinates matrix for the congruent framework $Q$
		\STATE Compute $X = - \frac{1}{2} J D J$.
		\STATE Do singular value decomposition on $X$ as $X = V \Lambda V^\top$, 
		where $V = [\mathbf{v_1}, \mathbf{v_2}, \mathbf{v_3}, \mathbf{v_4}, \mathbf{v_5}]$ is a $5 \times 5$ unitary matrix and $\Lambda$ is a diagonal matrix whose diagonal elements $\lambda_1 \geq \lambda_2 \geq \lambda_3 \geq \lambda_4 \geq \lambda_5\geq 0$ are singular values.
		\STATE Compute $\Lambda_* = \tM{diag}(\lambda_1,\lambda_2,\lambda_3) \in \mathbb R^{3 \times 3}$ and $V_* = [\mathbf{v_1}, \mathbf{v_2}, \mathbf{v_3}] \in \mathbb R^{5 \times 3}$.
		\STATE Compute $Q = \Lambda_*^{1/2} V_*^\top$.
		\RETURN Q
	\end{algorithmic}
\end{algorithm}

Once a congruent framework of the clique of order $5$ is generated using the CFC Algorithm, we compute the volumes of the tetrahedrons in the congruent framework using \eqref{Eq: volume}, and then obtain the barycentric coordinates $\{ a_{ij}, a_{ik}, a_{il}, a_{ih} \}$ of node $i$ using \eqref{Eq: barycentric coordinates - a}.

Ultimately, using \eqref{Eq: barycentric coordinates}, we establish the following linear equation for the Euclidean coordinates of the five nodes belonging to the clique:
\begin{equation}
\label{Eq: linear constraint - individual}
        \mathbf{p_i} = a_{ij} \mathbf{p_j} + a_{ik} \mathbf{p_k} + a_{il} \mathbf{p_l} + a_{ih} \mathbf{p_h},
\end{equation}
where the constants $a_{ij}$, $ a_{ik}$, $a_{il}$, and $a_{ih}$ are known, and $\mathbf{p_i}, \mathbf{p_j}, \mathbf{p_k}, \mathbf{p_l}, \mathbf{p_h}$ are the unknown absolute Euclidean coordinates of nodes $i, j, k, l$, and $h$ in the global coordinate system.

It is easy to find that a node may lie in multiple cliques. For any of them, we can establish one linear equality. Supposing that node $i$ lies in a total of $r_i$ cliques, we aggregate all the equations \eqref{Eq: linear constraint - individual} for the free nodes in matrix form as
\begin{equation} 
\label{Eq: linear constraint - matrix}
	% \begin{bmatrix}
	% 	\mathbf{p_1}\\ \vdots\\ \mathbf{p_1}\\ \vdots\\ \mathbf{p_{n-4}}\\ \vdots\\ \mathbf{p_{n-4}}
	% \end{bmatrix}
	\begin{bmatrix}
		 \underbrace{\mathbf{p_1} \dots \mathbf{p_1}}_{r_1} \dots \underbrace{\mathbf{p_{n-4}} \dots \mathbf{p_{n-4}}}_{r_{n-4}}
	\end{bmatrix}^\top\hspace{-.2cm}= \left( \begin{bmatrix}
		C & B \\
	\end{bmatrix}
	\otimes I_3 \right)
	\begin{bmatrix}
		\mathbf{p_f}\\ \mathbf{p_a}
	\end{bmatrix},
\end{equation}
where $C \in \mathbb{R}^{ \left(\sum_{i=1}^{n-4} r_i \right) \times \left( n-4 \right)}$ and $B \in \mathbb{R}^{\left( \sum_{i=1}^{n-4} r_i \right) \times 4}$ are two matrices that collects all the barycentric coordinates of the vertices associated with free nodes and anchor nodes, respectively. Define
\begin{equation}
    E: = [ \underbrace{\mathbf{e_1}, \cdots, \mathbf{e_1}}_{r_1}, \underbrace{\mathbf{e_2}, \cdots, \mathbf{e_2}}_{r_2}, \cdots, \underbrace{\mathbf{e_{n-4}}, \cdots, \mathbf{e_{n-4}}}_{r_{n-4}} ]^\top,
\end{equation}
as a $\left( \sum_{i=1}^{n-4} r_i \right) \times \left( n-4 \right)$ matrix. Let $\bar{C} = C \otimes I_3$, $\bar{B} = B \otimes I_3$ and $\bar{E} = E \otimes I_3$, we can write \eqref{Eq: linear constraint - matrix} in the following compact matrix form
\begin{equation}
\label{Eq: linear constraint - EC}
	\left( \bar{E} - \bar{C} \right) \mathbf{p_f} = \bar{B} \mathbf{p_a}.
\end{equation}

Define $M = E - C$ and $\bar{M} = \bar{E} - \bar{C}$. Then \eqref{Eq: linear constraint - EC} becomes
\begin{equation}
\label{Eq: linear constraint - M}
	\bar{M} \mathbf{p_f}= \bar{B} \mathbf{p_a},
\end{equation}
where $\bar{M} \in \mathbb{R}^{\left( 3 \sum_{i=1}^{n-4} r_i \right) \times 3 \left(n - 4\right)}$.

Next, subject to the linear equation constrains in \eqref{Eq: linear constraint - M}, we address the node localizability problem and present the necessary and sufficient conditions of node localizability.

The least squares method is considered to obtain the solution by minimizing the sum of the squares of the residuals made in the results of each individual equation. For the equation system \eqref{Eq: linear constraint - M}, the least squares formula is obtained from the problem
\begin{equation}
    \min_{\mathbf{p_f}} \| \bar{M} \mathbf{p_f} - \bar{B} \mathbf{p_a} \|,
\end{equation}
the solution of which can be written with the normal equation
\begin{equation}
\label{Eq: solution to least square formula}
	\mathbf{p_f} = \left( \bar{M}^\top \bar{M} \right)^{-1} \bar{M}^\top \bar{B} \mathbf{p_a},
\end{equation}
provided $ \left( \bar{M}^\top \bar{M} \right)^{-1}$ exists, which is equivalent to $\bar{M}$ has full column rank. Note that for \eqref{Eq: solution to least square formula}, an approximate solution is found when no exact solution exists. Moreover, \eqref{Eq: solution to least square formula} can be viewed as the solution to the following linear equation system
\begin{equation}
\label{Eq: linear constraint - MM}
	\bar{M}^\top \bar{M} \mathbf{p_f} = \bar{M}^\top \bar{B} \mathbf{p_a},
\end{equation}
if $\bar{M}^\top \bar{M}$ is invertible.

%Therefore, we propose an equivalent definition of node localizability based on the barycentric coordinates representation, expressed in the following lemma.
% \begin{lemma}
% 	A node $i\in\mathcal V$ is localizable if it belongs to a clique of order $5$ and $x_i = p_i$ for any $\mathbf{x}$ satisfying $\bar{M}^\top \bar{M} \mathbf{x} = \bar{M}^\top \bar{B} \mathbf{p_a}$, where $x_i$ is the $i$th component of $\mathbf{x}$.
% \end{lemma}

\begin{remark}
	It is easy to see that $\mathbf{p_f}$ in \eqref{Eq: linear constraint - MM} can be uniquely solved iff $\tM{rank} (\bar{M}^\top \bar{M}) = 3 \left( n - 4 \right)$, which is equivalent to $\tM{rank} (\bar{M}) = 3 \left( n - 4 \right)$ and $\tM{rank} (M) = n - 4$. Hence if the rank of the matrix $\bar{M}$ is less than $3 (n - 4)$, namely, the rank of the matrix $M$ is less than $n - 4$, there must exist unlocalizable nodes in the sensor network.
\end{remark}

By utilizing the definitions $\bar{M} = M \otimes I_3$ and $\bar{B} = B \otimes I_3$, we can interpret the linear equation system \eqref{Eq: linear constraint - MM} as three equivalent linear equation systems: $ M^\top M \mathbf{p_f^\varphi} = M^\top B \mathbf{p_a^\varphi}$, where the superscript $\varphi \in \{x, y, z \}$ denotes the Euclidean coordinate of $\varphi$ axis. Here, $\mathbf{p_f^\varphi}$ is a vector comprising the $\varphi$ axis Euclidean coordinates of all free nodes. Consequently, the localizability test for node $i$ can be transformed into a more concise exploration based on the aforementioned linear equations for $\varphi \in \{x, y, z \}$, which share the same matrix $M^\top M$. This dimension reduction also leads to a reduction in computational complexity during the verification test. Therefore, we can conclude that a node $i\in\mathcal V_f$ is localizable if it belongs to a clique of order $5$ and for any $\mathbf{x^\varphi}$ satisfying $M^\top M \mathbf{x^\varphi} = M^\top B \mathbf{p_a^\varphi} $, it must hold that $x_i^\varphi = p_i^\varphi$, where $x_i^\varphi$ and $p_i^\varphi$ are the corresponding scalars in $\mathbf{x_i} = [x_i^x, x_i^y, x_i^z]^\top$ and $\mathbf{p_i} = [p_i^x, p_i^y, p_i^z]^\top$, respectively.

Now we are ready to establish the necessary and sufficient conditions for localizability using the barycentric coordinates.

\begin{theorem} \label{Tho: Localizablity}
Supposing that node $i$ belongs to at least one clique of order $5$, then node $i$ is localizable iff $\mathbf{e_i} \bot \tM{ker}\left( M^\top M \right)$, where $\tM{ker}(M^\top M)$ denotes the kernel of the matrix $M^\top M$.
\end{theorem}

\begin{proof}
We start by proving sufficiency. Suppose to the contrary that node $i$ is not localizable, That is to say, there exists a vector $\mathbf{x^\varphi}$ that satisfies $M^\top M \mathbf{x^\varphi} = M^\top B \mathbf{p_a^\varphi}$ but does not satisfy $x_i^\varphi = p_i^\varphi$. Then it can be obtained that $M^\top M (\mathbf{x^\varphi} - \mathbf{p_f^\varphi}) = 0$. That is, $\mathbf{x^\varphi} - \mathbf{p_f^\varphi} \in \tM{Ker} (M^\top M)$ but the $i$th component is nonzero. Hence, $\mathbf{e_i}^\top \left( \mathbf{x^\varphi} - \mathbf{p_f^\varphi} \right) \neq 0$, which contradicts to the condition $\mathbf{e_i}^\top \perp \tM{ker}(M^\top M)$.

Now, we prove necessity. If node $i$ is localizable, then for any vectors $\mathbf{x^\varphi}$ and $\mathbf{y^\varphi}$ satisfying
\begin{equation}
\label{Eq: theorem proof - x}
M^\top M \mathbf{x^\varphi} = M^\top B \mathbf{p_a^\varphi}
\end{equation}
and
\begin{equation}
\label{Eq: theorem proof - y}
M^\top M \mathbf{y^\varphi} = M^\top B \mathbf{p_a^\varphi},
\end{equation}
it must hold that $x_i^\varphi = y_i^\varphi$, where $x_i^\varphi$ and $y_i^\varphi$ are the $i$th components of $\mathbf{x^\varphi}$ and $\mathbf{y^\varphi}$. That is to say, for any $w \in \tM{ker}(M^\top M)$, the $i$th component of $w$ must be zero as otherwise there must exist two different vectors $\mathbf{x^\varphi}$ and $\mathbf{y^\varphi}$ satisfying \eqref{Eq: theorem proof - x} and \eqref{Eq: theorem proof - y} such that $x_i^\varphi \neq y_i^\varphi $. Therefore, $\mathbf{e_i} \bot \tM{ker}\left( M^\top M \right)$.
\end{proof}

% \begin{remark}
% \color{red}Since the fact that the sensor network is localizable implies that all nodes in the sensor network are localizable, we can draw the following conclusion about network localizability. The sensor network is localizable if for each node in the sensor network belongs to at least one clique of order $5$ and $\tM{ker}(M^\top M) = 0$.
% \end{remark}

%%%%%%%
\subsection{Distributed Verification Algorithm for Node Localizability}

Building on Theorem~\ref{Tho: Localizablity}, we design a distributed algorithm to address Problem~\ref{problem1}, determining whether a generic node in the sensor network is localizable. In this algorithm, we compute the eigenvalues and eigenvectors of the matrix $M^\top M$ in a distributed fashion, following three steps. First, we design a subalgorithm to compute local sum of the node based on the barycentric coordinates only through neighbor sensing and communication. Second, we propose a subalgorithm that achieves sum consensus among all nodes in the sensor network within a finite number of iterations. Third, leveraging the aforementioned algorithms, design our distributed verification algorithm that solves Problem~\ref{problem1}.

\subsubsection{Local sum}
Given an arbitrary matrix $X$ with appropriate dimensions, where the $i$th row vector or column vector is denoted as $\mathbf{x_i}$, this Local Sum (LS) algorithm is designed to compute $\mathbf{y_i}$, the corresponding row or column vectors of the matrix products $M^\top M X$, $\bar{M}^\top \bar{M} X$, or $\bar{M}^\top \bar{B} X$, in a distributed manner. A pseudocode for this algorithm is reported in Subalgorithm~\ref{Alg: Local Sum Algorithm}.

In order to provide a clearer explanation of the algorithm, we recall that the matrix $M$ is structured in such a way that each row consists of the barycentric coordinates $a_{ij}$ of free nodes, i.e. $j \in \mathcal{N}_i \cap \mathcal{V}_f$, with the convention that $a_{ii} = -1$, and the remaining elements are zeros. While matrix $B$ only comprises the barycentric coordinates $a_{ij}$ of anchor nodes, i.e. $j \in \mathcal{N}_i \cap \mathcal{V}_a$. Note that the barycentric coordinates $a_{ij}$ can be computed using the congruent framework by Algorithm~\ref{Alg: Congruent Framework Construction Algorithm}. Exploiting these structures, the computation of each row in the matrix multiplication involving matrix $M$, $\Bar{M}$, and $\Bar{B}$ can be reformulated as a calculation of sums obtained by multiplying each $a_{ij}$ with the corresponding vector $\mathbf{x_i}$ held by neighboring nodes. This is possible due to the fact that the elements in each row of $M$, $\bar{M}$, and $\bar{B}$ only involve the corresponding node and its neighbors. Hence we name this computation as Local Sum and it can be efficiently performed by enabling communication solely between nodes and their respective neighbors. 

In this way, for each node $i$ that has the knowledge of $\mathbf{x_i}$, after obtaining $\mathbf{x_j}$ by communicating with neighboring nodes, the barycentric coordinates $a_{ij}$ can be utilized to calculate $\mathbf{y_i}$ using the LS algorithm in Subalgorithm~\ref{Alg: Local Sum Algorithm}.

On the one hand, the LS algorithm will be applied multiple times in the Distributed Localizablity Verification Algorithm. Given the $(n - 4)$-dimensional row vector $\mathbf{x_j}$ held by neighboring nodes of node $i$, it enables the local acquisition of $\mathbf{y_i}$ (the $i$th row vector of $M^\top M X$) through communication with neighbors. Here, $X$ is an arbitrary $(n - 4) \times (n - 4)$ matrix, with $\mathbf{x_j}$ as the $j$th row vector. The whole process could be described as follows. After receiving $\mathbf{x_j}$ through communication with neighbors, the algorithm computes an intermediate result, consisting with the row vectors of $\xi = M X$. Then, employing another round of communication through which nodes exchange this intermediate result, nodes are finally able to compute the desired $i$th row vector of $M^\top M X$, here denoted by $\mathbf{y_i}$. This step is key for the distributed verification of the relation \eqref{Eq: linear constraint - MM}, as illustrated at the end of this section.

On the other hand, the LS Algorithm will also be used later in this paper to pursue localization. In that application, the input $\mathbf{x_j}$ is a $3 \times 1$ column vector denoting the coordinate estimate of node $j$. As a result, both $\mathbf{\xi_i^r}$ and $\mathbf{y_i}$ are $3 \times 1$ column vectors. The LS Algorithm is then used to compute the components of $\mathbf{y} = \bar{M}^\top \bar{M} X$ or $\mathbf{y} = \bar{M}^\top \bar{B} X$, where $X$ is an arbitrary $3(n - 4) \times 1$ column vector consisting of several $3 \times 1$ vectors as its components, with the $j$th $3 \times 1$ vector as $\mathbf{x_j}$. Since $\bar{M}$ is composed of barycentric coordinates of free nodes, and $\bar{B}$ consists of barycentric coordinates of anchor nodes, we can differentiate the computation of $\bar{M}X$ or $\bar{B}X$ by selecting only free nodes or anchor nodes during the calculation of $\xi$ in line 2. More specifically, we compute the components of $\mathbf{\xi} = \bar{M} X$ when $j \in \{ \mathcal{N}_i \cup i\} \cap \mathcal{V}_f \cap \tetrahedron_i^r$ and $\mathbf{\xi} = \bar{B} X$ when $j \in \{ \mathcal{N}_i \cup i\} \cap \mathcal{V}_a \cap \tetrahedron_i^r$, which will be specially noted when calling this subalgorithm. Note that the superscript $r$ indicates that node $i$ belongs to $r$th clique of order 5, more details can be found in Section~\ref{SubSec: Node Localizability Condition}.

\floatname{algorithm}{Subalgorithm}
\begin{algorithm}[H]
	\caption{LS Algorithm}
	\label{Alg: Local Sum Algorithm}
	\renewcommand{\algorithmicrequire}{\textbf{$\bullet$}}
	\begin{algorithmic}[1]
		\REQUIRE \textbf{Function:} $\mathbf{y_i} = \textbf{LS}(\mathbf{x_j}, j \in \mathcal{N}_i)$
		\STATE Receive $\mathbf{x_j}$ from its neighbour $j \in \mathcal{N}_i$.
		\STATE Compute
		$$
		\begin{gathered}
			\vdots \\
			\mathbf{\xi_i^r} = \sum_{ j \in \{ \mathcal{N}_i \; \cup \; i\} \; \cap \; \tetrahedron_i^r } -a_{i j}^r \mathbf{x_j} \\
			\vdots
		\end{gathered}
		$$
		\STATE Receive $a_{j i}^r \mathbf{\xi_j^r}$ from its neighbour $j \in \mathcal{N}_i$.
		\STATE Compute
		$$
		\mathbf{y_i} = \sum_{ j \in \{ \mathcal{N}_i \; \cup \; i\} \; \cap \; \mathcal{V}_f, \forall r } - a_{j i}^r \mathbf{\xi_j^r}.
		$$
		\STATE \textbf{Return} $\mathbf{y_i}$.
	\end{algorithmic}
\end{algorithm}

\subsubsection{Finite-Time K-Max-Consensus Sum}

Next, we present the Finite-Time K-Max-Consensus Sum (FKMS) Algorithm, along with its associated function, the K-Max-Consensus (KMC) Algorithm. These algorithms are devised to achieve sum consensus among all nodes of the sensor network. Specifically, the algorithms enable each node to obtain the sum of all the state of the nodes in the network in a distributed fashion through communication. In these two algorithms, $ID_i$ and $x_i$ denote the identifier and state value of node $i$, and $\delta$ denotes the network diameter, which we assumed to be known.

\begin{remark}	
If the diameter $\delta$ is unknown, we can calculate it in a finite number of steps using established algorithms from the literature~\cite{almeida2012fast, peleg2012distributed, oliva2016distributed-b}, which only require an upper bound on the number of nodes that can also be computed in a distributed fashion (see, e.g.,~\cite{varagnolo2010distributed, varagnolo2013distributed, zhang2016distributed}).
\end{remark}

The KMC algorithm, described in the pseudocode in Subalgorithm~\ref{Alg: K-Max-Consensus Algorithm}, is utilized to obtain the $K$ maximum state values among all nodes in the sensor network and their corresponding identifiers. For each node, the entire process involves continuous comparison between its own $K$ maximum state values and those held by its neighbors, ultimately obtaining the $K$ maximum state values and corresponding identifiers. The key lies in the $\text{K}_{\max}$ function, which sorts all state values of the node and its neighbors by value first and then by identifier if there are repeated values. Following this procedure, since the network is connected, each node in the sensor network will finally select the same $K$ maximum values and corresponding identifiers. This iterative process requires $\delta$ steps, and the resulting values are stored in a local variable for each node, denoted by $\bar{\Omega}_i$, $i\in\mathcal V$.

\begin{algorithm}[H]
	\caption{KMC Algorithm}
	\label{Alg: K-Max-Consensus Algorithm}
	\renewcommand{\algorithmicrequire}{\textbf{$\bullet$}}
	\begin{algorithmic}[1]
		\REQUIRE \textbf{Function:} $ \bar{\Omega}_i = \textbf{KMC}\left(\left( ID_i, x_i \right), \delta \right)$
		\REQUIRE \textbf{Initialization:}
		\STATE $\Omega_i(0) = \{(ID_i,x_i)\}$.
	\end{algorithmic}
	\begin{algorithmic}[1]
		\REQUIRE \textbf{Iteration:}
		\FOR{$ k = 1, k \leq \delta, k++$}
		\STATE Receive $\Omega_j(k-1)$ from its neighbor $j \in \mathcal{N}_i$.
		\STATE $\Omega_i(k) = \textbf{K}_{\max} \cup_{ j \in \{ \mathcal{N}_i \; \cup \; i \} } \Omega_j(k-1)$.
		\ENDFOR
		\STATE \textbf{Return} $\bar{\Omega}_i = \Omega_i(k)$.
	\end{algorithmic}
\end{algorithm}

Then, the FKMS Algorithm (whose pseudocode is reported in Subalgorithm~\ref{Alg: Finite-time K-Max-Consensus Sum Algorithm}) uses the KMC Algorithm to enable each node to obtain the sum of all the state values held by all nodes. The overall process is described in the following three core steps of each iteration. First, each node updates its sum value by adding the $K$ maximum values. Second, any state value that has already been included in the sum is assigned a value equal to $- \infty$. Third, the KMC Algorithm is executed to identify the subsequent set of $K$ maximum values. These three steps are repeated iteratively until all state values eventually become equal to $- \infty$. This iterative process guarantees that each node eventually obtains the same sum value, which is referred to as sum consensus in this paper.

In Subalgorithm~\ref{Alg: Finite-time K-Max-Consensus Sum Algorithm}, $x_i^{tmp}$ and $sum_i$ denote the temporary variables of state value and sum of node $i$, respectively. The functions $\text{Keys}(\cdot)$ and $\text{Values}(\cdot)$ refer to the sets of identifiers and their corresponding state values. %, akin to the dictionary data structure in Python. 
The final consensus sum value is denoted by $\Bar{sum}$. We initialize $\bar{\Omega} = \{(ID_i,x_i)\}$ and $x_i^{tmp} = x_i$ for each node. During the iteration, on line 2 we call the Subalgorithm~\ref{Alg: K-Max-Consensus Algorithm} to get the $K$ maximum state values and their identifiers for $\delta$ steps; then, line 3 updates the sum value by adding the nonnegative infinity values of the $K$ maximum values once. In lines 4--5, we remove the nodes that have completed the accumulation of state values. At last, when all values $z \in \text{Values} \left( \bar{\Omega} \right)$ are $- \infty$, the algorithm terminates and every node obtains the consensus sum result. 

We want to stress that $K$ in Subalgorithm~\ref{Alg: Finite-time K-Max-Consensus Sum Algorithm} and Subalgorithm~\ref{Alg: K-Max-Consensus Algorithm} is a value that considers the tradeoff between memory size and iteration steps: the greater the value of $K$, the more memory it is used, but the less the steps are required. Therefore, $K$ can be selected according to the specific performance of sensors and demand of the application considered. 

\begin{algorithm}[H]
	\caption{FKMS Algorithm}
	\label{Alg: Finite-time K-Max-Consensus Sum Algorithm}
	\renewcommand{\algorithmicrequire}{\textbf{$\bullet$}}
	\begin{algorithmic}[1]
		\REQUIRE \textbf{Function:} $ \bar{sum} = \textbf{FKMS} \left( \left(ID_i, x_i \right), \delta, K \right) $
	\end{algorithmic}
	\begin{algorithmic}[1]
		\REQUIRE \textbf{Initialization:}
		\STATE $\bar{\Omega} = \{(ID_i,x_i)\}$.
		\STATE $x_i^{tmp} = x_i$.
		\STATE $sum_i = 0$.
	\end{algorithmic}
	\begin{algorithmic}[1]
		\REQUIRE \textbf{Iteration:} %with $k$ starting at $1$
		\WHILE{$z \in \text{Values} \left( \bar{\Omega} \right) > - \infty$}
		\STATE $\bar{\Omega} = \textbf{KMC}\left( \left(ID_i, x_i^{tmp} \right), \delta \right)$.
		\STATE $sum_i = sum_i + \sum_{z \in \text{Values} \left( \bar{\Omega} \right) \text{ and } z > - \infty} z $.
		\IF {$ID_i \in \text{Keys} \left( \bar{\Omega} \right)$ or $x_i^{tmp} = - \infty $}
		\STATE $x_i^{tmp} = - \infty$.
		\ENDIF
		\ENDWHILE
		\STATE \textbf{Return} $\bar{sum} = sum_i$.
	\end{algorithmic}
\end{algorithm}

\begin{remark}
The tradeoff between memory size and iteration steps through the parameter $K$ implemented in the FKMS Subalgorithm can also be encapsulated within other distributed consensus algorithms, including, e.g., average consensus~\cite{chapter_handbook}.
\end{remark}

\subsubsection{Distributed Localizablity Verification}

Finally, we present the Distributed Localizablity Verification, whose goal is to check the localizability of each node. In accordance with Theorem~\ref{Tho: Localizablity}, considering that eigenvectors corresponding to eigenvalues equal to $0$ reside in the kernel set of a matrix, we address this localizability verification problem by conducting distributed computations of the matrix eigenvalues and eigenvectors of matrix $M^\top M$. Specifically, our algorithm, whose pseudocode is reported in Algorithm~\ref{Alg: Distributed Localizablity Verification Algorithm}, computes for the eigenvalues of the matrix $M^\top M$ and calculates the $i$th component of each eigenvector at node $i$.

In Algorithm~\ref{Alg: Distributed Localizablity Verification Algorithm}, the superscript $\top$ indicates it is a row vector and subscript $s$ indicates the $s$th component of the vector, for example, $v_{i, s}^\top$ denotes the $s$th component of the row vector $\mathbf{v_i}^\top$. In the initialization, the initial value $\mathbf{v_i}^\top(0)$ of each node is set randomly. During the iteration, in line 1 we compute the components of $U = M^\top M V$ using the LS Algorithm, where matrix $V = \left[ \mathbf{v_1}, \cdots, \mathbf{v_{n-4}} \right]\top$ denotes the estimate of the matrix formed by the $n - 4$ eigenvectors of $M^\top M$ as columns. Then, in line 2, we use the FKMS Algorithm to compute the $n - 4$ eigenvalues using the Rayleigh quotient formula~\cite{horn2012matrix,parlett1998symmetric}. This step allows each node to compute the eigenvalues through a distributed computation process. In line 3, we obtain the consensus matrix $O = U^\top U$ by computing all of its elements, using again the FKMS Algorithm. Then, in lines 4 and 5, we perform the distributed orthonormalization of $U$, i.e., $U^\top U = (V R)^\top (V R) = R^\top R$, where $V$ and $R$ are the QR factorization of $U$. Once the eigenvalues are deemed converged, each node $i$ has access to $\mathbf{v_i}^\top$, which represents the $i$th component of all eigenvectors of $M^\top M$. This information is precisely what is needed to check the localizability verification condition for node $i$ from Theorem~\ref{Tho: Localizablity}. For practical applications, we set a maximum number of iterations ${\rm Iter}_{\max}$ and we reformulate the conditions as $\lambda_{s}({\rm Iter}_{\max}) < \varepsilon_{1}$ and $\left| v_{i,s}^\top({\rm Iter}_{\max}) \right| > \varepsilon_{2}$, where $\varepsilon_{1} > 0$ and $\varepsilon_{2} > 0$ are small threshold values. Fig.~\ref{Fig: flow chart} reports a flow chart of the algorithm.

\floatname{algorithm}{Algorithm}
\begin{algorithm}[!hbt]
	\caption{Distributed Localizablity Verification Algorithm}
	\label{Alg: Distributed Localizablity Verification Algorithm}
	\renewcommand{\algorithmicrequire}{\textbf{$\bullet$}}
	\begin{algorithmic}[1]
		\REQUIRE \textbf{Initialization:} $(\forall i = 1, \cdots, n-4)$ :
		\STATE $\mathbf{v_i}^\top(0) = $ random initial $(n - 4)$ dimensional row vector.
	\end{algorithmic}
	\begin{algorithmic}[1]
		\REQUIRE \textbf{Iteration:} 
		with $k$ starting at 0 :
		%\FOR{$k = 0, k++$}
		\STATE $\mathbf{u_i}^\top(k) = $ \textbf{LS} $\left( \mathbf{v_j}^\top(k), j \in \{ \mathcal{N}_i \cup i\} \cap \mathcal{V}_f\right)$.
		\STATE $\lambda_s(k) = \frac{\textbf{FKMS} \left( \left(ID_i,v_{i, s}^\top(k) u_{i, s}^\top(k) \right), \delta, K \right)}{\textbf{FKMS} \left( \left(ID_i, v_{i, s}^\top(k) v_{i, s}^\top(k) \right), \delta, K\right)}, \forall s = 1, \cdots, n-4$.
		\STATE $O(k)= \textbf{FKMS} \left( \left(ID_i, \mathbf{u_i}(k) \mathbf{u_i}^\top(k) \right), \delta, K \right)$.
		\STATE Solve $R(k)$ from $R^\top(k) R(k) = O(k)$. %through Cholesky factorization 
		\STATE $\mathbf{v_i}^\top(k+1) = \mathbf{u_i}^\top(k) R^{-1}(k)$.
		\STATE $k = k+1$.
		%\ENDFOR
	\end{algorithmic}
	\begin{algorithmic}[1]
		\REQUIRE \textbf{Localizability verification:}
		%\IF {$\exists s$ such that $\lambda_{s}(\rm Iter}_{\max}) = 0$ but $v_{i, s}^H(\rm Iter}_{\max}) \neq 0$,}
		\IF {$\exists s$ such that $\lambda_{s}({\rm Iter}_{\max}) < \varepsilon_{1}$ but $\left| v_{i, s}^\top({\rm Iter}_{\max}) \right| > \varepsilon_{2}$,}
		\STATE node $i$ is unlocalizable;
		\ELSE
		\STATE node $i$ is localizable.
		\ENDIF
	\end{algorithmic}
\end{algorithm}

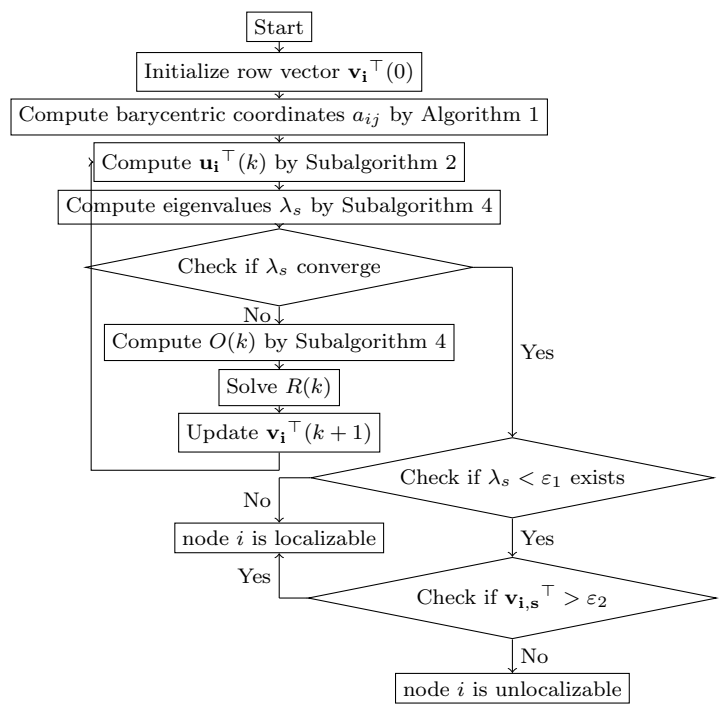
\begin{figure}[!htbp]
\centering
\scriptsize
\tikzstyle{format}=[rectangle, draw, thin, fill = white]
\tikzstyle{test}=[diamond, aspect=5, draw, thin]
\tikzstyle{point}=[coordinate, on grid]
\begin{tikzpicture}
	% Before Check Converge
	\node[format] (start){Start};
	\node[format, below of = start, node distance = 6mm] (Initialize v){Initialize row vector $\mathbf{v_i}^\top(0)$};
	\draw[->] (start)--(Initialize v);
	\node[format, below of = Initialize v, node distance = 6mm] (compute aij){Compute barycentric coordinates $a_{ij}$ by Algorithm~\ref{Alg: Congruent Framework Construction Algorithm}};
	\draw[->](Initialize v)--(compute aij);
	\node[format, below of = compute aij, node distance = 6mm] (compute u){Compute $\mathbf{u_i}^\top(k)$ by Subalgorithm~\ref{Alg: Local Sum Algorithm}};
	\draw[->](compute aij)--(compute u);
	\node[format, below of = compute u, node distance = 6mm] (compute lambda){Compute eigenvalues $\lambda_{s}$ by Subalgorithm~\ref{Alg: Finite-time K-Max-Consensus Sum Algorithm}};
	\draw[->](compute u)--(compute lambda);
 
	% Check Converge
	\node[test, below of = compute lambda, node distance = 8mm](check converge){Check if $\lambda_{s}$ converge};
	\draw[->](compute lambda)--(check converge);
 
	% Check Converge No
	\node[format, below of = check converge, node distance = 10mm](compute O){Compute $O(k)$ by Subalgorithm~\ref{Alg: Finite-time K-Max-Consensus Sum Algorithm}};
	\draw[->](check converge) --node[left]{No} (compute O);
	\node[format, below of = compute O, node distance = 6mm](solve R){Solve $R(k)$};
	\draw[->](compute O)--(solve R);
	\node[format, below of = solve R, node distance = 6mm](update v){Update $\mathbf{v_i}^\top(k+1)$};
	\draw[->](solve R)--(update v);
	\node[point, below of = update v, node distance = 5mm](point below of update v){};
	\draw[-](update v)--(point below of update v);
	\node[point, left of = compute u, node distance = 25mm](point left of compute u){};
	\draw[-](point below of update v)-|(point left of compute u);
	\draw[->](point left of compute u)--(compute u.west);

	% Check Converge Yes
	% Check Lambda
	\node[point, right of = check converge, node distance =31mm](point right of check converge){};
	\draw[-](check converge)--(point right of check converge);
	\node[test, below of = point right of check converge, node distance = 28mm](check lambda){Check if $\lambda_{s} < \varepsilon_{1}$ exists};
	\draw[->](point right of check converge)--node[right]{Yes}(check lambda);
	
	% Check Lambda No
	\node[point, left of = check lambda, node distance = 31mm](point left of check lambda){};
	\draw[-](check lambda)--(point left of check lambda);
	\node[format, below of = point left of check lambda, node distance = 8mm](localizable){node $i$ is localizable};
	\draw[->](point left of check lambda)--node[left]{No}(localizable);
	
	% Check Lambda Yes
	% Check V
	\node[test, below of = check lambda, node distance = 16mm](check v){Check if $\mathbf{v_{i,s}}^\top > \varepsilon_{2}$};
	\draw[->](check lambda)--node[right]{Yes}(check v);
	
	% Check V Yes
	\node[point, left of = check v, node distance = 31mm](point left of check v){};
	\draw[-](check v)--(point left of check v);
	\draw[->](point left of check v)--node[left]{Yes}(localizable);
	
	% Check V No
	\node[format, below of = check v, node distance = 12mm](unlocalizable){node $i$ is unlocalizable};
	\draw[->](check v)--node[right]{No}(unlocalizable);
\end{tikzpicture}
\caption{Flow chart of the Distributed Localizablity Verification Algorithm. }
\label{Fig: flow chart}
\end{figure}

\begin{remark}
	We want to stress that the Distributed Localizability Verification Algorithm with the KMC and FKMS Algorithms can be used not only to solve Problem~\ref{problem1} in this setting, but their applicability can also be extended to all localization problems employing barycentric coordinates; namely, in two or three dimensional space with various types of measurements. For instance, it can be adopted before running the localization algorithms in~\cite{diao2014barycentric, cheng2016single, han2017barycentric} for the sensor network with unlocalizable nodes. The key lies in modifying the LS Algorithm according to the barycentric coordinates obtained in the specific scenario considered. Thus, our Distributed Localizability Verification Algorithm is a general distributed algorithm for testing the localizability of each node in the sensor network using barycentric coordinates.
\end{remark}

%——————————%%%%%%%——————————
\section{Finite-Time Localization Algorithm}  \label{Sec: Localization}

%%%%%%%
\subsection{Distributed Localization Algorithm}

Using the Distributed Localizablity Verification Algorithm in Algorithm~\ref{Alg: Distributed Localizablity Verification Algorithm}, each node is able to determine whether it is localizable or not. Then, if node $i$ is not localizable, all cliques containing node $i$ in the sensor network are removed, as well as the corresponding linear equations. At this stage, we define $\mathcal{V}_z$ as the set of $n_z$ localizable free nodes, forming a sub-graph $\mathcal G_z=(\mathcal V_z,\mathcal E_z)$ with $\mathcal{E}_z:={(i,j)\in\mathcal{E}: i,j\in\mathcal{V}_z}$. Additionally, we make the following assumption.
\begin{assumption}\label{a:connected}
    The sub-graph $\mathcal G_z=(\mathcal V_z,\mathcal E_z)$ is connected.
\end{assumption}
Let $\delta_z$ be the diameter of this sub-graph. Under Assumption~\ref{a:connected}, we observe that the following inequalities naturally hold:
\begin{equation}
\label{Eq: ineq}
    n_z\leq n-4<n,\quad \delta_z\leq n_z-1<n.
\end{equation}

The reduced set of linear equations from \eqref{Eq: linear constraint - MM} that involve only nodes belonging to $\mathcal{V}_z$ can be written as follows:
\begin{equation}
\label{Eq: linear constraint - MzMz}
	\bar{M}_z^\top \bar{M}_z \mathbf{p_z} = \bar{M}_z^\top \bar{B}_z \mathbf{p_a},
\end{equation}
where the subscript $z$ denotes the matrix or vector established by localizable nodes and $\mathbf{p_z}$ denotes the Euclidean coordinate vectors of the localizable node set $\mathcal{V}_z$.

Note that the matrix $\bar{M}_z^\top \bar{M}_z$ is symmetric and positive definite since all unlocalizable nodes are removed from the sensor network. Therefore, we propose our distributed localization algorithm to solve \eqref{Eq: linear constraint - MzMz} by leveraging the conjugate gradient method~\cite{hestenes1952methods}. In our Distributed Localization Algorithm (whose pseudocode is reported in Algorithm~\ref{Alg: Distributed Localization Algorithm}), each node obtains an initial value $r_i(0)$, which is equal to the $i$th component of
\begin{equation}
\label{Eq: r(0)}
    \mathbf{r}(0) \! 
= \! \bar{M}_z^\top \bar{B}_z \mathbf{p_a} \! - \! \bar{M}_z^\top \bar{M}_z \mathbf{\hat{p}_z}(0) \! 
= \! - \bar{M}_z^\top 
\hspace{-.1cm}\begin{bmatrix} \bar{M}_z \! - \! \bar{B}_z \end{bmatrix} \hspace{-.1cm}\begin{bmatrix} \mathbf{\hat{p}_z}(0) \\ \mathbf{p_a} \end{bmatrix}\hspace{-.1cm},
\end{equation}
where $\mathbf{\hat{p}_z}(0)$ denotes the initial estimate of the Euclidean coordinates of localizable nodes. During the iteration, in line 1 we use the LS Algorithm to compute $q_i(k)$ for node $i$, which is the $i$th components of $\mathbf{q}(k) = \bar{M}_z^\top \bar{M}_z \mathbf{v}(k)$. Then, in lines 2 and 5, we update the step sizes $\alpha_i(k)$ and $\beta_i(k)$. Note that $\alpha_i(k)$ and $\beta_i(k)$ are definitely constants consistent for all nodes, so they can be uniformly denoted as $\alpha(k)$ and $\beta(k)$ and computed as
\begin{equation}
\label{Eq: alpha and beta}
    \alpha(k) = \frac{\mathbf{r}^\top(k) \mathbf{r}(k)}{\mathbf{v}^\top(k) \mathbf{q}(k)},
\quad
\beta(k) = \frac{\mathbf{r}^\top(k+1) \mathbf{r}(k+1)}{\mathbf{r}^\top(k) \mathbf{r}(k)},
\end{equation}
in a distributed fashion by means of the FKMS Algorithm. Meanwhile, in lines 3, 4, and 6, we update the estimate $\mathbf{\hat{p}_i}(k)$ and parameters $\mathbf{r_i}(k)$ and $\mathbf{v_i}(k)$. Note that, the first two quantities are updated after computing the step $\alpha_i(k)$ since they require the updated values of $\alpha_i(k)$, and before computing the step $\beta_i(k)$, for which they are needed. Similar, the latter quantity ($\mathbf{v_i}(k)$) is updated at the end of the iteration, using the updated value of the step $\beta_i(k)$, computed in line 5.

\begin{algorithm}[H]
	\caption{Distributed Localization Algorithm}
	\label{Alg: Distributed Localization Algorithm}
	\renewcommand{\algorithmicrequire}{\textbf{$\bullet$}}
	\renewcommand{\algorithmicensure}{\textbf{Output:}}
	\begin{algorithmic}[1]
		\REQUIRE \textbf{Initialization:} $(\forall i \in \mathcal{V}_z)$
		\STATE $\mathbf{\hat{p}_i}(0) = $ random initial $3$ dimensional column vector.
		\STATE $\mathbf{\hat{p}_{n-3}}(0) = \mathbf{p_{n-3}}, \mathbf{\hat{p}_{n-2}}(0) = \mathbf{p_{n-2}}, \mathbf{\hat{p}_{n-1}}(0) = \mathbf{p_{n-1}}$ and $\mathbf{\hat{p}_n}(0) = \mathbf{p_n}$.
		\STATE $\mathbf{r_i}(0) = - \textbf{LS} \left( \mathbf{\hat{p}_j}(0), j \in \{ \mathcal{N}_i \cup i\} \cap \left( \mathcal{V}_z \cup \mathcal{V}_a \right) \right)$.
		\STATE $\mathbf{v_i}(0) = \mathbf{r_i}(0)$.
	\end{algorithmic}
	\begin{algorithmic}[1]
		\REQUIRE \textbf{Iteration:} with $k$ starting at $0$
		\STATE $\mathbf{q_i}(k) = $ \textbf{LS} $\left(\mathbf{v_j}(k), j \in \{ \mathcal{N}_i \cup i\} \cap \mathcal{V}_z\right)$.
		\STATE $\alpha_i(k) = \frac{\textbf{FKMS} \left( \left(ID_i, \mathbf{r_i}^\top(k) \mathbf{r_i}(k) \right), \delta_z, K \right)}{\textbf{FKMS} \left( \left(ID_i, \mathbf{v_i}^\top(k) \mathbf{q_i}(k) \right), \delta_z, K \right)}$.
		\STATE $\mathbf{\hat{p}_i}(k+1) = \mathbf{\hat{p}_i}(k) + \alpha_i(k) \mathbf{v_i}(k)$.
		\STATE $\mathbf{r_i}(k+1) = \mathbf{r_i}(k)-\alpha_i(k) \mathbf{q_i}(k)$.
		\STATE $\beta_i(k) = \frac{\textbf{FKMS} \left( \left(ID_i, \mathbf{r_i}^\top(k+1) \mathbf{r_i}(k+1) \right), \delta_z, K \right)}{\textbf{FKMS} \left( \left(ID_i, \mathbf{r_i}^\top(k) \mathbf{r_i}(k) \right), \delta_z, K \right)}$.
		\STATE $\mathbf{v_i}(k+1) = \mathbf{r_i}(k+1) + \beta_i(k) \mathbf{v_i}(k)$.
		\STATE $k = k+1$.
	\end{algorithmic}
\end{algorithm}

%%%%%%%
\subsection{Finite-Time Convergence Analysis}

We prove that the Distributed Localization Algorithm in Algorithm~\ref{Alg: Distributed Localization Algorithm} solves Problem~\ref{problem2}, i.e., that the positions estimated by the algorithm converge to the real positions in a finite number of steps. The analysis consists of two parts. First, we prove that the FKMS Algorithm can obtain the step sizes $\alpha$ and $\beta$ in finite time. Second, using this result, we prove finite-time convergence of Algorithm~\ref{Alg: Distributed Localization Algorithm}.

\subsubsection{Finite-time computation of the step size}

We start by proving the following lemma.

\begin{lemma} \label{Lem: FKMS - Convergence}
	Under Assumption~\ref{a:connected}, for any $i \in \mathcal V_z$ in the sensor network, $\bar{sum}_i$ generated by the FKMS Algorithm in Subalgorithm~\ref{Alg: Finite-time K-Max-Consensus Sum Algorithm} converges to the sum of the state values in finite time. Precisely, the algorithm takes no more than $\delta_z \left( \left\lceil \frac{n_z}{K} \right\rceil + 1 \right)$ steps to converge.%, where $\delta_z$ is the diameter of the sensor network. 
\end{lemma}

\begin{proof}
Since the sub-graph is connected by Assumption~\ref{a:connected}, it is easy to observe that the KMC Algorithm succeeds within $\delta_z$ steps during the first iteration and the $K$ maximum initial state values and their corresponding identifiers are obtained by all nodes. We denote them as $\bar{\Omega} = \textbf{KMC} \left( \left( ID_i, x_i(0) \right), \delta_z \right)$. Note that if state values are repeated, the largest identifiers are selected. After adding the $K$ maximum initial state values to $sum_i$, the algorithm sets the temporary state values $x_i^{tmp}$ of all nodes whose identifier is contained in $\text{Keys} \left( \bar{\Omega} \right)$ to $-\infty$, aiming to eliminate their value from subsequent KMC computations. Similar, in the following iterations, the sets $\bar{\Omega} = \textbf{KMC} \left( \left( ID_i, x_i^{tmp} \right) , \delta_z \right)$ are computed within $\delta_z$ steps each, where $\text{Values} \left( \bar{\Omega} \right) $ are the $K$ maximum values of the initial state values except previous sets of $K$ maximum values and $\text{Keys} \left( \bar{\Omega} \right)$ are their identifiers. 

Hence, in each iteration, each node adds $K$ maximum initial state values to $sum_i$. Once all initial state values have been accumulated, we have $x_i^{tmp} = - \infty$ for all nodes and $\text{Values} \left( \bar{\Omega} \right) = \{ -\infty, \cdots, -\infty \}$. At this time, the algorithm terminates simultaneously for all nodes. This requires a number of iterations equal to the smallest integer that is larger than or equal to $\frac{n_z}{K}$. Furthermore, an additional step is needed to ensure that all the $x_i^{tmp}$ values are equal to $- \infty$. Therefore, a total of $\left\lceil \frac{n_z}{K} \right\rceil + 1$ steps of the main routine of Subalgorithm~\ref{Alg: Finite-time K-Max-Consensus Sum Algorithm} are required. Moreover, it should be noted that every node has the same value for $sum_i$ at each step due to the consensus properties.

In conclusion, the KMC Algorithm requires $\delta_z$ steps for each iteration and the FKMS Algorithm needs to be executed no more than $\left\lceil \frac{n_z}{K} \right\rceil + 1$ times, yielding the claim. 
\end{proof}

In the light of Lemma~\ref{Lem: FKMS - Convergence}, we can draw a conclusion that for any node $i \in \mathcal{V}_z$, the numerator and denominator of $\alpha_i$ and $\beta_i$ can be calculated simultaneously in at most $\delta_z \left( \left\lceil \frac{n_z}{K} \right\rceil + 1 \right)$ steps. Therefore, the step sizes $\alpha$ and $\beta$ can be obtained in a finite number of steps, and the final results are shown as \eqref{Eq: alpha and beta}.

\subsubsection{Finite-time convergence of Algorithm~\ref{Alg: Distributed Localization Algorithm}}

After ensuring that the step size is obtained within a finite number of steps, we show that the Distributed Localization Algorithm converges in a finite number of iterations.

First, for all nodes $i$, $i \in \mathcal{V}_z$, the discrete-time system in Algorithm~\ref{Alg: Distributed Localization Algorithm} can be written in the following compact matrix 
\begin{subequations} \label{Eq: conjugate gradient}
	\begin{align}
		& \mathbf{\hat{p}_z}(k+1) = \mathbf{\hat{p}_z}(k) + \alpha(k) \mathbf{v}(k), \label{Eq: conjugate gradient-c} \\
		& \mathbf{r}(k+1) = \mathbf{r}(k) - \alpha(k) \bar{M}_z^\top \bar{M}_z \mathbf{v}(k), \label{Eq: conjugate gradient-d} \\
		& \mathbf{v}(k+1) = \mathbf{r}(k+1) + \beta(k) \mathbf{v}(k), \label{Eq: conjugate gradient-f} \\
		%\end{align}
		%\end{subequations}
		\intertext{with step sizes}
		%\begin{align}
		& \alpha(k) = \frac{\mathbf{r}^\top(k) \mathbf{r}(k)}{\mathbf{v}^\top(k) \bar{M}_z^\top \bar{M}_z \mathbf{v}(k)}, \label{Eq: conjugate gradient-b} \\
		& \beta(k) = \frac{\mathbf{r}^\top(k+1) \mathbf{r}(k+1)}{\mathbf{r}^\top(k) \mathbf{r}(k)}, \label{Eq: conjugate gradient-e}
	\end{align}
\end{subequations}
by substituting $\mathbf{q}(k) = \bar{M}_z^\top \bar{M}_z \mathbf{v}(k)$ into \eqref{Eq: alpha and beta}. In the sequel, we denote $A = \bar{M}_z^\top \bar{M}_z$.

Then, before providing the proof, we present a lemma that demonstrates the orthogonality of the residuals and conjugacy of the search direction, considering~\cite{hestenes1952methods}.

\begin{lemma} \label{Lem: Orthogonal Conjugate}
    The residuals $\mathbf{r}(0), \mathbf{r}(1), \dots$ generated by \eqref{Eq: conjugate gradient-d} are mutually orthogonal and the direction vectors $\mathbf{v}(0), \mathbf{v}(1), \dots$ generated by \eqref{Eq: conjugate gradient-f} are mutually conjugate with respect to the matrix $A = \bar{M}_z^\top \bar{M}_z$. 
\end{lemma}

\begin{proof}

The proof of the lemma is equivalent to prove that the following relations hold
\begin{subequations} \label{Eq: relations}
	\begin{align}
		\mathbf{r}^\top(g) \mathbf{r}(f) & = 0 \; (g \neq f), \label{Eq: relations-a} \\ 
		\mathbf{v}^\top(g) A \mathbf{v}(f) & = 0 \; (g \neq f), \label{Eq: relations-b} \\ 
		\mathbf{v}^\top(g) \mathbf{r}(f) = 0 \; (g < f), 
		& \; \mathbf{v}^\top(g) \mathbf{r}(f) = \mathbf{r}^\top(g) \mathbf{r}(g), \label{Eq: relations-c} \\ 
%		v^\top(g) r(f) &= 0 \quad (g < f), \nonumber \\
%		v^\top(g) r(f) & = r^\top(g) r(g), \label{Eq: relations-c} \\ 
%		r^\top(g) A v(g) & = v^\top(g) A v(g), \nonumber \\ 
%		r^\top(g) A v(f) & = 0 \quad (g \neq l, g \neq f + 1), \label{Eq: relations-d} \\
		\mathbf{r}^\top\!(g) A \mathbf{v}(g) \! = \! \mathbf{v}^\top\!(g) A \mathbf{v}(g),
		& \; \mathbf{r}^\top\!(g) A \mathbf{v}(f) \! = \! 0 \; ( \! g \! \neq \! l, g \! \neq \! f \! + \! 1 \! ), \label{Eq: relations-d}
	\end{align}
\end{subequations}
which could be made by induction. First, it easy to get that the vectors $\mathbf{r}(0), \mathbf{v}(0)$ and $\mathbf{r}(1)$ satisfy relations \eqref{Eq: relations} due to
\begin{equation}
\begin{aligned}
	\mathbf{r}&^\top(0) \mathbf{r}(1) = \mathbf{v}^\top(0) \mathbf{r}(1) = \mathbf{v}^\top(0) \left( \mathbf{r}(0) - \alpha(0) A \mathbf{v}(0) \right) \\
	&= \mathbf{r}^\top(0) \mathbf{r}(0) - \alpha(0) \mathbf{v}^\top(0) A \mathbf{v}(0) \\
	& = \mathbf{r}^\top(0) \mathbf{r}(0) - \frac{\mathbf{r}^\top(0) \mathbf{r}(0)}{\mathbf{v}^\top(0) A \mathbf{v}(0)} \left(\mathbf{v}^\top(0) A \mathbf{v}(0)\right) = 0,
\end{aligned}
\end{equation}
%$$
%r^\top(0) r(1) = v^\top(0) r(1) = v^\top(0) \left( r(0) - \alpha(0) A v(0) \right)
%$$
%$$
%= r^\top(0) r(0) - \alpha(0) v^\top(0) A v(0)
%$$
%$$
%= r^\top(0) r(0) - \frac{r^\top(0) r(0)}{v^\top(0) A v(0)} \left(v^\top(0) A v(0)\right) = 0
%$$
by \eqref{Eq: conjugate gradient-d} and \eqref{Eq: conjugate gradient-b}. 

Second, supposing that relations \eqref{Eq: relations} hold for the vectors $\mathbf{r}(0), \cdots, \mathbf{r}(k)$ and $\mathbf{v}(0), \cdots, \mathbf{v}(k-1)$, we prove that $\mathbf{v}(k)$ satisfies relations \eqref{Eq: relations}. To verify that $\mathbf{v}(k)$ can be adjoined to this set, it is necessary to show that
\begin{subequations} \label{Eq: relations - k}
	\begin{align}
		\mathbf{r}^\top(g) \mathbf{v}(k) = \mathbf{r}^\top(k) \mathbf{r}(k) &\quad (g \leq k), \label{Eq: relations - k-a} \\
		%	& \small{\rightarrow \eqref{Eq: relations-c}-2} 
		\mathbf{v}^\top(g) A \mathbf{v}(k) = 0 &\quad (g < k), \label{Eq: relations - k-b} \\
		%	& \small{\rightarrow \eqref{Eq: relations-b}}
		\mathbf{r}^\top(k) A \mathbf{v}(g) = \mathbf{v}^\top(k) A \mathbf{v}(g) &\quad (g \leq k, g \neq k - 1). \label{Eq: relations - k-c} 
		%	& \small{\rightarrow \eqref{Eq: relations-d}-1} 
	\end{align}
\end{subequations}
It can be verified that \eqref{Eq: conjugate gradient-e} and \eqref{Eq: conjugate gradient-f} hold for any $k$ iff
\begin{equation}
\label{Eq: vk}
	\mathbf{v}(k) = \left( \mathbf{r}^\top(k) \mathbf{r}(k) \right) \sum\nolimits_{f=0}^k \frac{\mathbf{r}(f)}{\mathbf{r}^\top(f) \mathbf{r}(f)} . %\quad k = 0, 1, 2, \cdots
\end{equation}
Then for formula~\ref{Eq: relations - k}, \eqref{Eq: relations - k-a} can be derived at once from \eqref{Eq: vk} and \eqref{Eq: relations-a} as follows
\begin{equation}
\begin{array}{l}
    \mathbf{r}^\top(g) \mathbf{v}(k) 
= \mathbf{r}^\top(g) \left( \mathbf{r}^\top(k) \mathbf{r}(k) \sum_{l = 0}^k \frac{\mathbf{r}(f)}{ \mathbf{r}^\top(f) \mathbf{r}(f)}\right)\\
\,\,\,= \mathbf{r}^\top(g) \left(\mathbf{r}^\top(k) \mathbf{r}(k) \frac{\mathbf{r}(g)}{\mathbf{r}^\top(g) \mathbf{r}(g)} \right) 
= \mathbf{r}^\top(k) \mathbf{r}(k).
\end{array}
\end{equation}
To prove \eqref{Eq: relations - k-b}, we use \eqref{Eq: conjugate gradient-d} and find that
\begin{equation}
\begin{array}{l}
\mathbf{r}^\top(g+1) \mathbf{v}(k) = \left( \mathbf{r}(g) - \alpha(g) A \mathbf{v}(g) \right)^\top \mathbf{v}(k)\\
\qquad= \mathbf{r}^\top(g) \mathbf{v}(k) - \alpha(g) \mathbf{v}^\top(g) A^\top \mathbf{v}(k),
\end{array}
\end{equation}
which becomes
\begin{equation}
    \mathbf{r}^\top(k) \mathbf{r}(k) = \mathbf{r}^\top(k) \mathbf{r}(k) - \alpha(g) \mathbf{v}^\top(g) A \mathbf{v}(k)
\end{equation}
by \eqref{Eq: relations - k-a} when $g < k$, then \eqref{Eq: relations - k-b} holds since $\alpha(g) > 0$. In order to establish \eqref{Eq: relations - k-c}, we use \eqref{Eq: conjugate gradient-f} and \eqref{Eq: relations - k-b} to obtain
\begin{equation}
\begin{array}{l}
\mathbf{v}^\top(k) A \mathbf{v}(g) = \left( \mathbf{r}(k) + \beta(k-1) \mathbf{v}(k-1) \right)^\top A \mathbf{v}(g) \\
\qquad= \mathbf{r}^\top(k) A \mathbf{v}(g) + \beta(k-1) \mathbf{v}^\top(k-1) A \mathbf{v}(g) 
\\\qquad= \mathbf{r}^\top(k) A \mathbf{v}(g),
\end{array}
\end{equation}
when $g \neq k - 1$, which follows that \eqref{Eq: relations - k-c} holds. Therefore, relations \eqref{Eq: relations} holds for the vectors $\mathbf{r}(0), \mathbf{r}(1), \cdots, \mathbf{r}(k)$ and $\mathbf{v}(0), \mathbf{v}(1), \cdots, \mathbf{v}(k)$.

Finally, supposing that relations \eqref{Eq: relations} holds for the vectors $\mathbf{r}(0), \cdots, \mathbf{r}(k)$ and $\mathbf{v}(0), \cdots, \mathbf{v}(k - 1)$, we prove $\mathbf{r}(k + 1)$ satisfies relations \eqref{Eq: relations}. With the above results, proving that $\mathbf{r}(k + 1)$ can be adjoined to this set is done by showing that
\begin{subequations} 
\label{Eq: relations - k+1}
	\begin{align}
		\mathbf{r}^\top(g) \mathbf{r}(k+1) & = 0 \quad (g \leq k), \label{Eq: relations - k+1-a} \\
		%	& \small{\rightarrow \eqref{Eq: relations-a}}
		\mathbf{v}^\top(g) A^\top \mathbf{r}(k+1) & = 0 \quad (g < k), \label{Eq: relations - k+1-b} \\
		%	& \small{\rightarrow \eqref{Eq: relations-d}-2}
		\mathbf{v}^\top(g) \mathbf{r}(k+1) & = 0 \quad (g \leq k). \label{Eq: relations - k+1-c}	
		%	& \small{\rightarrow \eqref{Eq: relations-c}-1}
	\end{align}
\end{subequations}
By \eqref{Eq: conjugate gradient-d}, we have
\begin{equation}
\label{Eq: term0}
\mathbf{r}^\top(g) \mathbf{r}(k+1) = \mathbf{r}^\top(g) \mathbf{r}(k) - \alpha(k) \mathbf{r}^\top(g) A \mathbf{v}(k).
\end{equation}
When $g < k$, the terms on the right of \eqref{Eq: term0} are both $0$ and then \eqref{Eq: relations - k+1-a} holds. When $g = k$, the right member is still zero
\begin{equation}
\begin{array}{l}
\mathbf{r}^\top(k) \mathbf{r}(k) - \alpha(k) \mathbf{r}^\top(k) A \mathbf{v}(k) \\
\quad= \mathbf{r}^\top(k) \mathbf{r}(k) - \frac{\mathbf{r}^\top(k) \mathbf{r}(k)}{\mathbf{v}^\top(k) A \mathbf{v}(k)} \left(\mathbf{r}^\top(k) A \mathbf{v}(k)\right)\\
\quad= \mathbf{r}^\top(k) \mathbf{r}(k) - \mathbf{r}^\top(k) \mathbf{r}(k) = 0,
\end{array}
\end{equation}
by \eqref{Eq: conjugate gradient-b} and \eqref{Eq: relations-d}. Therefore, \eqref{Eq: relations - k+1-a} holds. Besides, when $g < k$, using \eqref{Eq: conjugate gradient-d} again, we have
\begin{equation}
\begin{array}{l}
0 = \mathbf{r}^\top(k+1)\mathbf{r}(g+1) = \mathbf{r}^\top(k+1)\mathbf{r}(g) - \\
\quad\alpha(g) \mathbf{r}^\top(k+1)A \mathbf{v}(g)= - \alpha(g) \mathbf{r}^\top(k+1)A \mathbf{v}(g),
\end{array}
\end{equation}
hence \eqref{Eq: relations - k+1-b} holds. The equation \eqref{Eq: relations - k+1-c} follows from \eqref{Eq: relations - k+1-a} and the formula \eqref{Eq: vk} for $v(g)$, if $g \leq k$, we have
\begin{equation}
    \mathbf{v}^\top(g) \mathbf{r}(k+1) \! = \! \Bigg( \! \left(\mathbf{r}^\top(g) \mathbf{r}(g)\right) \! \sum_{f=0}^g \frac{\mathbf{r}(f)}{\mathbf{r}^\top(f) \mathbf{r}(f)} \Bigg)^\top \! \mathbf{r}(k+1) \! = \! 0.
\end{equation}
Therefore, \eqref{Eq: relations} hold for the vectors $\mathbf{r}(0), \mathbf{r}(1), \cdots, \mathbf{r}(k), \mathbf{r}(k+1)$ and $\mathbf{v}(0), \mathbf{v}(1), \cdots, \mathbf{v}(k)$.
\end{proof}

Now, based on Lemma~\ref{Lem: Orthogonal Conjugate}, we prove that the system \eqref{Eq: conjugate gradient} converges in $m$ iterations, in the sense that the estimate $\mathbf{\hat{p}_z}$ coincides with the exact solution $\mathbf{p_z}$ of the system of linear equations \eqref{Eq: linear constraint - MzMz} after a finite number of iterations.

\begin{lemma} \label{Lem: CG - Convergence}
	For any $\mathbf{\hat{p}_z}(0) \in \mathbb{R}^{3 n_z}$, consider the sequence $\left\{\mathbf{\hat{p}_z}(k)\right\}$, $k = 1,2,\dots$ generated by \eqref{Eq: conjugate gradient}. Then, there exists a constant $m \leq 3n_z$ such that $\hat{p}_z(k) = \mathbf{p_z}$, solution of the linear equation system \eqref{Eq: linear constraint - MzMz}, for any $k \geq m$.
\end{lemma}

\begin{proof}
Let $m$ be the smallest integer such that the difference between $\mathbf{\hat{p}_z}(0)$ and the solution $\mathbf{p_z}$ is in the subspace spanned by $\mathbf{v}(0), \mathbf{v}(1), \cdots, \mathbf{v}(m-1)$. Due to the fact that $A = \bar{M}_z^\top \bar{M}_z \in \mathbb{R}^{3 n_z \times 3 n_z}$ and $\mathbf{\hat{p}_z} \in \mathbb{R}^{3 n_z}$, the dimension of the subspace must be less than or equal to $3n_z$. Hence, $m \leq 3n_z$. Since the vectors $\mathbf{v}(0), \mathbf{v}(1), \cdots, \mathbf{v}(m-1)$ are linearly independent according to Lemma~\ref{Lem: Orthogonal Conjugate}, we choose scalars $\mu(0), \cdots, \mu(m-1)$ such that we can write
% \begin{equation}
%     \mathbf{p_z} - \mathbf{\hat{p}_z}(0) = \mu(0) \mathbf{v}(0) + \cdots + \mu(m-1) \mathbf{v}(m-1),
% \end{equation}
% that is,
\begin{equation}
\label{Eq: pz}
\mathbf{p_z} = \mathbf{\hat{p}_z}(0) + \mu(0) \mathbf{v}(0) + \cdots + \mu(m-1) \mathbf{v}(m-1).
\end{equation}
Then, by substituting \eqref{Eq: pz} into \eqref{Eq: r(0)}, we get 
\begin{equation}
\begin{array}{lll}
    \mathbf{r}(0) &=& b - A \mathbf{\hat{p}_z}(0) = A (\mathbf{p_z} - \mathbf{\hat{p}_z}(0))\\&=& \mu(0) A \mathbf{v}(0) +\cdots + \mu(m-1) A \mathbf{v}(m-1),
\end{array}
\end{equation}
with $b = \bar{M}_z^\top \bar{B}_z \mathbf{p_a}$. Next, using Lemma~\ref{Lem: Orthogonal Conjugate}, we compute
\begin{equation}
\begin{array}{l}
\mathbf{v}^\top(k) \mathbf{r}(0) = \mathbf{v}^\top(k)\big( \mu(0) A \mathbf{v}(0) + \cdots + \mu(m-1)\cdot\\ \qquad A \mathbf{v}(m-1) \big) = \mu(k) \mathbf{v}^\top(k) A \mathbf{v}(k),
\end{array}
\end{equation}
%using Lemma~\ref{Lem: Orthogonal Conjugate},% \eqref{Eq: relations-b}, 
and then obtain 
\begin{equation}
    \mu(k) = \frac{\mathbf{v}^\top(k) \mathbf{r}(0)}{\mathbf{v}^\top(k) A \mathbf{v}(k)} = \frac{\mathbf{r}^\top(k) \mathbf{r}(k)}{\mathbf{v}^\top(k) A \mathbf{v}(k)}.
\end{equation}
%using Lemma~\ref{Lem: Orthogonal Conjugate}. %\eqref{Eq: relations-c}. 
So far, we find that $\alpha(k) = \mu(k)$, and hence that $\mathbf{p_z} = \mathbf{\hat{p}_m}$, which yields the claim. 
\end{proof}

Finally, we prove that Algorithm~\ref{Alg: Distributed Localization Algorithm} solves Problem~\ref{problem2}.

\begin{theorem} \label{Tho: Localization - Convergence}
	For any $\mathbf{\hat{p}_i}(0) \in \mathbb{R}^{3}$, $i \in \mathcal{V}_z$, the estimate $\mathbf{\hat{p}_i}(k)$ generated by the Distributed Localization Algorithm in Algorithm~\ref{Alg: Distributed Localization Algorithm} converges to the absolute position $\mathbf{p_i}$ in finite time. Specifically, it takes no more than $3n_z \left(2 \delta_z \left( \left\lceil \frac{n_z}{K} \right\rceil + 1 \right)\right)$ iterations to converge.
\end{theorem}

\begin{proof}
According to Lemma~\ref{Lem: FKMS - Convergence}, the step sizes $\alpha_i$ and $\beta_i$ can be obtained in at most $\delta_z \left( \left\lceil \frac{n_z}{K} \right\rceil + 1 \right)$ iterations, Therefore, running the iteration in Algorithm~\ref{Alg: Distributed Localization Algorithm} once requires a total of $2 \delta_z \left( \left\lceil \frac{n_z}{K} \right\rceil + 1 \right) $ iterations. Then, by applying the results of Lemma~\ref{Lem: CG - Convergence}, we conclude that the bound on the convergence of Algorithm~\ref{Alg: Distributed Localization Algorithm} is given by $3n_z \left(2 \delta_z \left( \left\lceil \frac{n_z}{K} \right\rceil + 1 \right) \right)$ iterations.
\end{proof}

\begin{corollary}
We can establish a conservative bound on the number of steps needed for convergence of the Distributed Localization Algorithm in terms of the total number of nodes of the network $n$, as being less than $3n \left(2n \left( \left\lceil \frac{n}{K} \right\rceil + 1 \right) \right)$. 
\end{corollary}

\begin{proof}
    This comes directly from Theorem~\ref{Tho: Localization - Convergence} and \eqref{Eq: ineq}.
\end{proof}

\begin{remark}
Here, we provide a theoretical bound on the number of steps in Algorithm~\ref{Alg: Distributed Localization Algorithm}. However, in practical scenarios, due to the presence of the rounding errors, it is possible that a few additional steps may be needed.
\end{remark}

%——————————%%%%%%%——————————
\section{Simulations}  \label{Sec: Simulation}

To validate our algorithms and test their effectiveness, we present a set of numerical simulations performed on two case studies of sensor networks of different size.

%%%%%%%
\subsection{Case Study I}

In the first case study, we consider a sensor network with $n=50$ nodes consisting of four anchor nodes (selected at random) and $46$ free nodes. Each node of the network is located in the three-dimensional space with each side of length $100$m. Inspired by practical implementations~\cite{niculescu2003ad,zhang2006secure,vetelino2017introduction}, we establish that sensors can communicate and take relative measurements if they are within a certain radius. Here, we set such a critical radius at the value of $50$m. Hence we generate the edges of the sensing and communication topology according to this rule, that is, $(i,j)\in \mathcal E\iff \|\mathbf{p_i}-\mathbf{p_j}\| \leq 50$m.

The configuration and the sensing and communication topology of the sensor network are shown in Fig.~\ref{Fig: topology-50nodes}. Using the filtering process designed in Algorithm~\ref{Alg: Distributed Localizablity Verification Algorithm}, we identified three unlocalizable nodes in the sensor network (in red in Fig.~\ref{Fig: topology-50nodes}), which are excluded from the localization process.

\begin{figure}
	\centering
	\includegraphics[width=\linewidth]{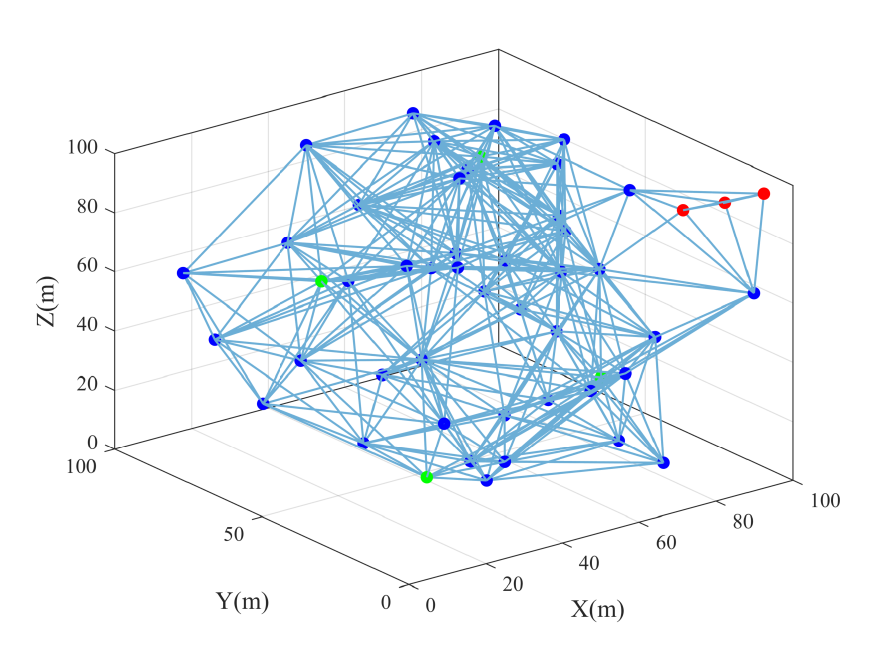}
	\caption{Configuration and topology of the sensor network for Case Study~I. Anchor nodes are denoted in green, unlocalizable nodes (detected by Algorithm~\ref{Alg: Distributed Localizablity Verification Algorithm}) are in red, and localizable nodes in blue. }\vspace{-10pt}
	\label{Fig: topology-50nodes}
\end{figure}

Then, we run Algorithm~\ref{Alg: Distributed Localization Algorithm} to estimate the positions of the remaining $43$ localizable free nodes, starting from initial estimates picked uniformly at random in the domain $100 \text{m}\times 100\text{m} \times 100$m. Fig.~\ref{Fig: convergence-50nodes-CG} illustrates the estimated positions at the end of the iterations (in red), i.e., after 5,369 steps (performed in approximately 18s on a 2.5 GHz 8-Core Intel Core i7-11700 computer), compared to the exact positions (in blue). Consistent with our theoretical guarantees, the estimates for all localizable nodes converge to their real positions. We also illustrate a sample trajectory of coordinate estimates for a single node in Fig.~\ref{Fig: convergence-50nodes-CG}. The trajectory shows how the estimate rapidly converges to the real position of the node, demonstrating the performance of our algorithm.

\begin{figure}
	\centering
	\includegraphics[width=\linewidth]{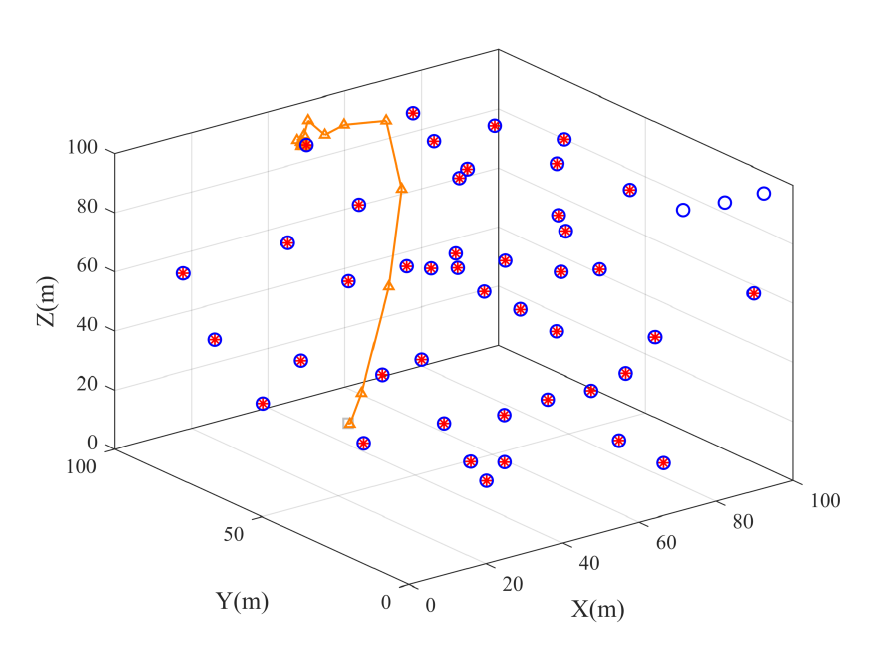}
	\caption{Output of our localization algorithm compared to the absolute positions for Case Study I. The blue circles denote the real absolute position and the red asterisks denote the final estimated position. The plot also depicts a sample trajectory: the initial estimate is represented by a grey square, the estimates at each 5 iterations by orange triangles. }\vspace{-10pt}
	\label{Fig: convergence-50nodes-CG}
\end{figure}

Now, we perform a comparison of our distributed localization algorithm (Algorithm~\ref{Alg: Distributed Localization Algorithm}) with the state-of-the-art algorithms proposed in the literature. In particular, we consider the Jacobi Under Relaxation Iteration algorithm (JU) from~\cite{xia2022exploratory} and the Richardson Iteration algorithm (RI), proposed in~\cite{diao2014barycentric, cheng2016single, han2017barycentric}. In Fig.~\ref{Fig: error-50nodes}, we report the results of our comparison. In particular, the plot shows how the estimation error ratio $\|\mathbf{\hat{p}_z}(k)-\mathbf{p_z}\| / \|\mathbf{\hat{p}_z}(0)-\mathbf{p_z}\|$ obtained using our Algorithm~\ref{Alg: Distributed Localization Algorithm} (denoted as CG) compared to the same quantity computed for the other two algorithms under the same settings (real positions and initial estimates). From this comparison, we note that, after a short transient, the distributed localization algorithm proposed in this paper outperforms the other two, making a significant enhancement in the convergence rate. Moreover, it eventually allows to determine the exact locations of all nodes in a finite number of steps as guaranteed by Theorem~\ref{Tho: Localization - Convergence}. 
%For 50nodes: The network diameter is $4$ and we choose $K = 5$ in this paper, $\alpha = 0.5$ in JU and $\varepsilon = 2 / \left( \min( \text{eig} \left( \bar{M}_z^\top \bar{M}_z) \right) + \max \left( \text{eig}(\bar{M}_z^\top \bar{M}_z) \right) \right)$ in RI.

\begin{figure}
	\centering
	\includegraphics[width=\linewidth]{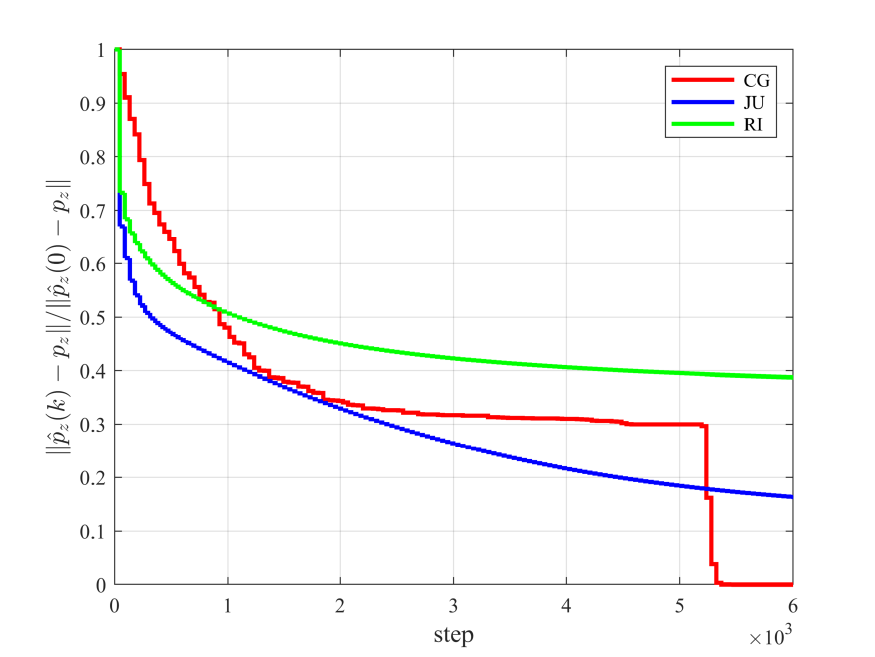}
	\caption{Evolution of the estimation error ratio using Algorithm~\ref{Alg: Distributed Localization Algorithm} (red), compared to JU~\cite{xia2022exploratory} (blue) and RI algorithms~\cite{han2017barycentric} (green) for Case Study~I.}\vspace{-10pt}
	\label{Fig: error-50nodes}
\end{figure}

%%%%%%%
\subsection{Case Study II}

Finally, we perform a set of simulations on a larger-scale network to verify the scalability of our algorithms. Specifically, we consider a scenario in which $n=1,000$ sensors are placed in the same $100\text{m}\times100\text{m}\times100$m three-dimensional space, with their positions set at random, and edges added if two sensors are within $20$m of radius. As shown in Fig.~\ref{Fig: convergence-1000nodes-CG}, Algorithm~\ref{Alg: Distributed Localizablity Verification Algorithm} can successfully identify $2$ unlocalizable nodes and, then, Algorithm~\ref{Alg: Distributed Localization Algorithm} allows to exactly determine the absolute positions of all the remaining $994$ free nodes in finite time. 

\begin{figure}
\centering
	\includegraphics[width=\linewidth]{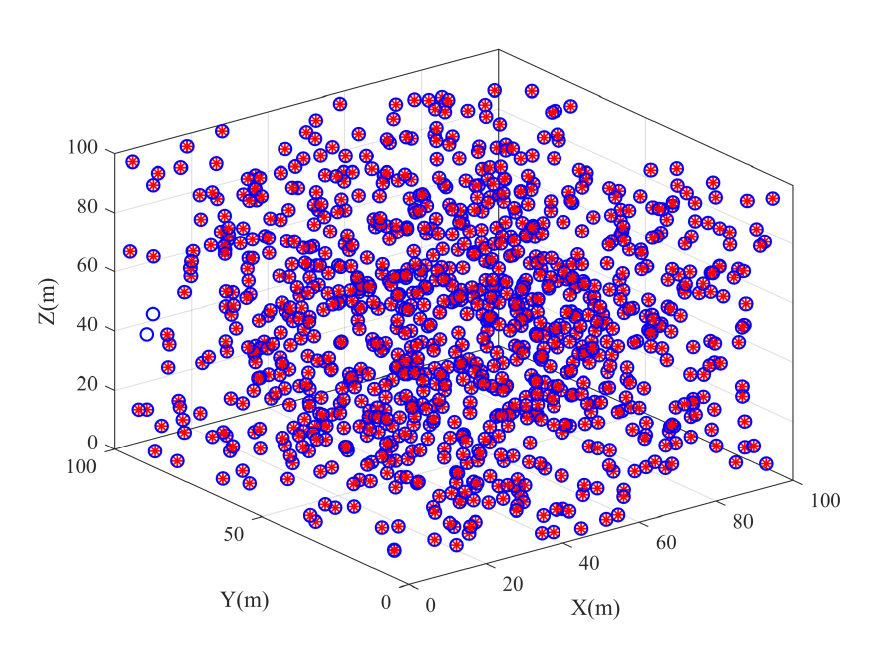}
	\caption{Output of our localization algorithm compared to the absolute positions for Case Study II. The blue circles denote the real absolute position and the red asterisks denote the final estimated position. }\vspace{-10pt}
	\label{Fig: convergence-1000nodes-CG}
\end{figure}

Similar to the previous case study, we compared the estimation error ratio of our algorithm with the other two algorithms from the state-of-the-art literature. The results, reported in Fig.~\ref{Fig: error-1000nodes}, show that as the number of nodes and network diameter increase, the performances of Algorithm~\ref{Alg: Distributed Localization Algorithm} scales well. In fact, even though computational effort needed to estimate the exact positions increases, we observe that the convergence steps remain order of magnitude smaller than the one needed for the other two algorithms. In addition, the number of steps can also be reduced at the cost of allocating more memory to the solver by increasing the value of the parameter $K$ in the FKMS Algorithm, through the tradeoff between memory size and iteration steps discussed at the end of Section~\ref{Sec: Localization}.

\begin{figure}
	\centering
	\includegraphics[width=\linewidth]{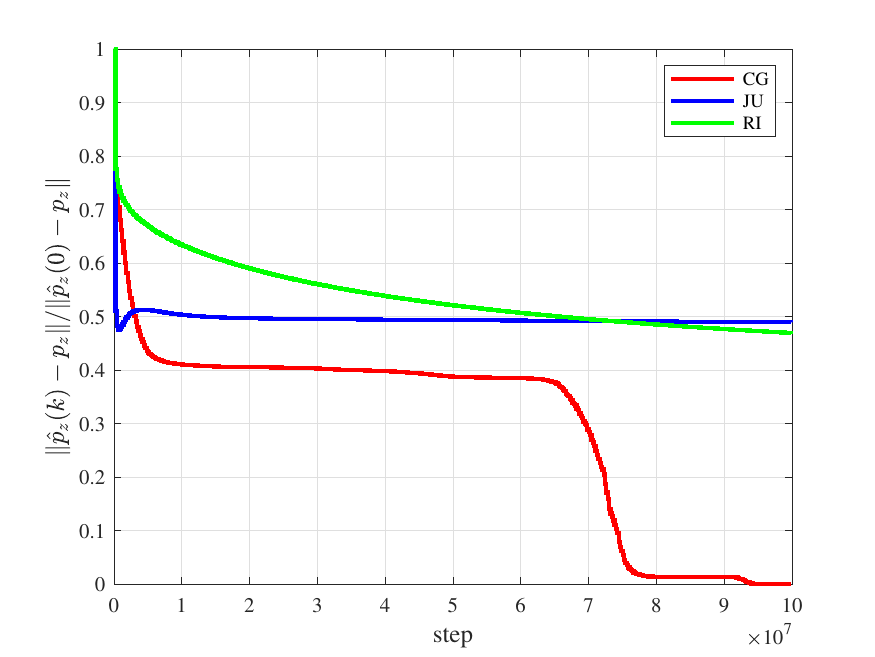}
	\caption{Evolution of the estimation error ratio using our algorithm (red), compared to JU~\cite{xia2022exploratory} (blue) and RI algorithms~\cite{han2017barycentric} (green) for Case Study II.}\vspace{-10pt}
	\label{Fig: error-1000nodes}
\end{figure}

%——————————%%%%%%%——————————
\section{Conclusion}  \label{Sec: Conclusion}

We studied the distributed localization problem for three-dimensional sensor networks with range measurements. First, utilizing barycentric coordinates, we proposed a distributed localizability verification algorithm to identify which nodes are unlocalizable. Second, building on a conjugate gradient method, we proposed an efficient distributed localization algorithm, which is able to determine the location of all localizable nodes in finite time. Third, numerical simulations were offered to demonstrate the performance of our algorithm compared to the state-of-the-art.

The results presented in this paper pave the way for several avenues of future research. First, real-world measurements and information exchange are often subject to noises. The study of the performance of our distributed algorithms in the presence of noises is of paramount importance for assessing their applicability in real-world scenarios. Second, an interesting future research direction is to improve the performance of our algorithm, designing simultaneous localizability verification and localization. Third, to extend the applicability of our methods, our algorithms should be augmented with the design of ad-hoc methods to determine the positions of unlocalizable sensors. Fourth, the methodologies presented in this paper can be extended to different localization settings. In particular, in the process of solving the localization problem we designed a general finite-time sum consensus algorithm. Such an algorithm can be employed to effectively solve other problems, in which a consensus output should be computed in finite time.

\section*{Acknowledgment}
This work was partially supported by National Natural Science Foundation of China, under Grant No. 62173118; Shenzhen Key Laboratory of Control Theory and Intelligent Systems, under grant No. ZDSYS20220330161800001; and FAIR --- Future Artificial Intelligence Research, and received funding from the European Union Next-GenerationEU (PIANO NAZIONALE DI RIPRESA E RESILIENZA (PNRR) --- MISSIONE 4 COMPONENTE 2, INVESTIMENTO 1.3 – D.D. 1555 11/10/2022, PE00000013).

\end{document}